\documentclass[11pt,twoside]{amsart}
\usepackage{txfonts}
\usepackage{bbm}
\usepackage{mathrsfs}
\usepackage{amsfonts}
\usepackage{amsmath}
\usepackage{amssymb}
\usepackage{wasysym}
\usepackage{psfrag}
\usepackage{graphics,epsfig,amsmath}  
\usepackage{comment}
\usepackage{enumerate}
\newtheorem{Definition}{Definition}
\newtheorem{Proposition}{Proposition}
\newtheorem{Lemma}{Lemma}
\newtheorem{Corollary}{Corollary}
\newtheorem{Remark}{Remark}

\def\Tau {\mathcal{T}}

\title[Aubry-Mather theory] {Percival Lagrangian approach
to  Aubry-Mather theory}
\author[X. Su]{Xifeng Su}
\address{Dept. of Mathematics, Nanjing University, Nanjing, 210093, CHINA}
\address{Dept. of Mathematics, Univ. of Texas at Austin,
1 University Station C1200, Austin TX 78712-0257}
\email{billy3492@gmail.com, xifengsu@math.utexas.edu}
\author[R. de la Llave]{Rafael de la Llave}
\address{Dept. of Mathematics, Univ. of Texas at Austin,
1 University Station C1200, Austin TX 78712-0257}
\email{llave@math.utexas.edu}

\begin{document}
\maketitle

\begin{abstract}
We present some streamlined proofs of some of the basic results in
Aubry-Mather theory (existence of quasi-periodic minimizers,
multiplicity results when there are gaps among minimizers)
 based on the study of hull functions.
We present results in arbitrary number of dimensions

We also compare the proofs and results with those obtained in other
formalisms.

\end{abstract}

\keywords{ Quasi-periodic solutions, variational methods,
Aubry-Mather theory, hull functions}

\subjclass[2000]{
70H12, 
82B20, 
37A60, 
49J40, 
39A14 
}

\section{Introduction}
Many problems in dynamics and in solid state physics lead to the
study of minimizers and other critical points of (formal)
variational problems. One wants to establish existence and geometric
properties of these minimizers and critical points.

For example, orbits of a twist map are critical points of the action
(see \cite{Gole'01}). In other physical problems (e.g. motion of
dislocations, spin waves, etc.) the interpretation of the
variational principle is energy and the critical points are
equilibrium states (see \cite{Aubry'83,Monneau'08}), whereas
minimizers are ground states.

The theory of critical points for such functionals  was studied by
mathematicians  very intensely since the early 80's due to the
systematic work of Aubry (\cite{Aubry'83}) and Mather
(\cite{Mather'82}), (but there are precedents in the mathematical
work of Morse and Hedlund in the 30's
\cite{Morse'24,Morse'73,MR1503086} and much more work by physicists
\cite{MR2035039}).

{From} the point of view of analysis,  one of the problems of the
theory is that the variational problems are formal and that
therefore, one cannot use a straightforward approach to the calculus
of variations. Also, to look for quasi-periodic solutions, one has
to deal with functionals in
$\ell^\infty=\{~\{u_i\}_{i\in\mathbb{Z}^d}~|~\|u\|_{\ell^\infty}\equiv\sup_{i\in\mathbb{Z}^d}|u_i|<\infty\}$
which is a notoriously ill-behaved space.

For example, we will be dealing with the variational problem for
\emph{``configurations''} i.e $u:\mathbb{Z}^d \rightarrow \mathbb{R}
$
 \begin{equation}\label{lagrangian}
    \mathscr{L}(u)=\sum_{i\in\mathbb{Z}^d}\sum_{j=1}^d
    H_{j}(u_{i},u_{i+e_j})
\end{equation}

The case $d = 1$ corresponds to twist mappings. When we are looking
for quasi-periodic solutions, the sums in \eqref{lagrangian} are
clearly, not meant to converge but there are ways of associating
well defined variational problems to the formal functionals
\eqref{lagrangian}.

There are many standard ways of dealing with such problems. The two
main ones are: A) To  work in spaces of sequences defining precisely
what one means by  minimizers, critical values of the action, etc.
This is what was done in the classical calculus of variations
starting with \cite{Morse'24}. B) We assume that $u$ are
parameterized by a function $h$ -- the hull function -- and a
frequency $\omega \in \mathbb{R}^d$ such that
\begin{equation} \label{hullform}
u_i = h(\omega \cdot i)
\end{equation}
and derive a variational principle for $h$.

\subsection{Heuristic derivation of the Percival
Lagrangian}\label{heuristic}
The heuristic derivation of the variational principle in B) is as
follows \cite{Percival'79}. If we assume solutions of the form
\eqref{hullform}, considering a big box and normalizing the
Lagrangian (which does not change the minima or critical points), we
are led to considering
 \begin{equation*}
    \mathscr{L}_{N,\omega}(u)=
\frac{1}{N^d} \sum_{i\in\mathbb{Z}^d, |i| \le N}\sum_{j=1}^d
    H_{j}(h(\omega \cdot i), h(\omega\cdot i + \omega_j))
\end{equation*}
Heuristically, for
$N\rightarrow\infty,~\mathscr{L}_{N,\omega}\rightarrow\mathscr{P}_\omega$
where
\begin{equation}\label{Percivalian}
   \mathscr{P}_\omega(h) =
   \sum_{j=1}^d \int_0^{1} H_{j}(h(\theta), h(\theta +\omega_j))
\end{equation}
This heuristic derivation shows that given a solution $h$ of
$\omega$'s of the form \eqref{hullform}, $\mathscr{P}_\omega(h)$ has
a direct physical interpretation as the energy per volume.

In a similar heuristic way, we can argue that the Euler-Lagrange
equations for $\mathscr{P}_\omega$ are obtained by computing
\[
\begin{split}
&\mathscr{P}_\omega( h +\varepsilon \eta)  - \mathscr{P}_\omega(h)\\
=& \varepsilon  \int_0^1 d\theta \sum_j
 \partial_1 H_j ( h(\theta), h(\theta + \omega_j)) \eta(\theta)
+ \partial_1 H_j ( h(\theta), h(\theta + \omega_j)) \eta(\theta +
\omega_j)
+ O(\varepsilon^2)  \\
=& \int_0^1 d \theta\, \big[ \sum_j \partial_1 H_j ( h(\theta),
h(\theta + \omega_j)) + \partial_2 H_j( h(\theta - \omega_j) ,
h(\theta)) \big] \eta(\theta) + O(\varepsilon^2)
\end{split}
\]
If $\eta$ is arbitrary, the Euler-Lagrange equations should be:

\begin{equation}\label{EulerLagrangePercival}
X(h)\equiv
 \sum_j \partial_1 H_j ( h(\theta), h(\theta + \omega_j))
+ \partial_2 H_j( h(\theta - \omega_j) , h(\theta)) = 0
\end{equation}
Of course, the above heuristic derivation is rather imprecise since,
depending on the space of $h$'s we consider the variations allowed
may not be arbitrary and minimizers may not satisfy Euler-Lagrange
equation.

Besides the heuristic derivation, the paper \cite{Percival'79} used
this formalism as a very effective numerical method to compute
quasi-periodic solutions.

We also note that this formalism can be used as the basis of KAM
theory to produce smooth solutions under some assumptions (
Diophantine properties of the frequencies that the system is close
to integrable, etc.) (see \cite{MR982563,MR2507323})

The  rigorous study of \eqref{Percivalian} that we will pursue  here
entails
\begin{itemize}
\item{I)}  To identify appropriate spaces in which one can
study $\mathscr{P}_\omega$ and show that it has a minimizer
satisfying geometric properties.
\item{II)} To show that the minimizer of $\mathscr{P}_\omega$
satisfies the Euler-Lagrange equations.
\item{III)} To show  that the minimizers thus obtained,
correspond to minimizers in the formalism A)
\item{IV)} To show existence of other critical points and
their properties provided there are two minimizers that are
essentially different.
\end{itemize}

We point out that there are different tradeoffs. If we choose a very
restrictive space on which to consider the minimization problem
(i.e. a space of functions enjoying many properties), then it
becomes hard to show that the minimizer exists and that it satisfies
the Euler-Lagrange equations. On the other hand, if we choose very
general spaces, the minimizers may become useless. One has to
consider spaces general enough so that minimizers exist and satisfy
Euler-Lagrange equations, but restrictive enough so that they
satisfy enough properties that can be bootstrapped. This compromise
is, of course, far from unique and we will make a point of showing
several such compromises.

Another point to keep in mind is that the problem of existence of
minimizers can be approximated by simpler ones (under rather soft
assumptions the limit of minimizers of a sequence of problems is a
minimizer of the limiting problem \cite{Morse'73}). On the other
hand, passing to the limit on multiplicity results is difficult
because the limits of two different solutions of the approximating
problems could be the same.

Even if the results that we will obtain have already been obtained
(in the $d=1$ case), the spaces that we choose are different and we
obtain shorter proofs. For the proof in IV) we use a gradient flow
approach.

We refer to \cite{Bangert'88,Forni'96, Gole'01} for surveys of the
classical results. Notably, the results of existence of minimizers
were obtained by the hull function approach in \cite{Mather'82}, the
critical points in \cite{Mather'86}. Besides the fact that we deal
with $d > 1$ and more general lattices, we think it is worthwhile to
present the arguments in a coherent way.  It is also interesting to
compare the approach presented in this paper with that in
\cite{Rafael'07a} which covers similar ground (it includes weak
twist, and long range interactions) using methods based on orbit
spaces.

Of course, Aubry-Mather theory has grown well beyond the results
that we consider here and now includes  studies of other  objects
such as Mather measures, Ma\~n\'e critical values, which lead to
applications to construction of connecting orbits, viscosity
solutions, transport theory, multi-bump solutions etc. Some surveys
on these more recent aspects are \cite{Manebook,Mane'96b,
Contreras'99, Fathi'97, Figalli}.

\begin{Remark}
It is very important to note that the variational principles are
degenerate in the sense that minimizers will not be unique. We note
that if $u_i$ is a ground state  (resp. a critical point) of
$\mathscr{L}$,  then, for any $k\in\mathbb{Z}^d$ so is $\tilde u$
defined by $\tilde u_i = u_{i+k}$.

Similarly, if $u$ is a minimizer (resp. a critical point) of
$\mathscr{P}_\omega$, so is $\tilde h$ defined by $\tilde h(\theta)
= h(\theta + a)$.

Note that, in the hull function formalism, the symmetries are
continuous symmetries indexed by the real number $a$ whereas in the
orbit formalism, the symmetries are indexed by $k\in\mathbb{Z}^d$.

Later on, we will see that our assumptions on $H$ will imply other
symmetries of the problem.
\end{Remark}

\begin{Remark}
We note that the methods we consider can be extended with only
typographical changes in  the formulas to interactions of infinite
range and involving many bodies
\begin{equation}\label{lagrangianextended}
    \mathscr{L}(u)= \sum_{L \in \mathbb{N}}
\sum_{i\in\mathbb{Z}^d}
    H_{L}(\Tau_i u)
\end{equation}
where $H_{L}(u) $ depends only  $\{u_j\}_{|j| \le L}$ and $\Tau_i$
is the translation. Of course, one needs to assume that the
interactions decrease fast enough with the distance $L$.

It is easy to see that when we consider \eqref{lagrangianextended},
the corresponding variational principle for the hull functions
\begin{equation} \label{percivalianextended}
    \mathscr{P}_\omega(h)=
\int_0^1 d\theta \sum_{L \in \mathbb{N}}
    H_{L}(h)(\theta)
\end{equation}
where $H_L(h)$ is the function obtained replacing $u_j$ by $h(\theta
+ \omega\cdot j) $.
\end{Remark}



\begin{Remark}
In Appendix A, we will show how the method of hull functions can be
extended to study configurations $u:\Lambda\rightarrow \mathbb{R}$
when $\Lambda$ is , e.g. the Bethe lattice.

This is somewhat surprising because the heuristic derivation
outlined in Section ~\ref{heuristic} uses that $\mathbb{Z}^d$ is
amenable and the Bethe lattice is not amenable.
\end{Remark}

\begin{Remark}
Note that in point $III)$ we show that the minimizers of the hull
function approach give rise to minimizing sequences, which also
satisfy some growth properties at $\infty$ and several monotonicity
properties.

It seems to be an open question to decide whether there are
converses to $III)$. That is, whether the minimizers satisfying
several order and growth properties are of the form
\eqref{hullform}, in particular, they depend on just one variable.

There are several versions of these questions formulated for PDE's
in \cite{Moser'86,Bangert'89}. In \cite{JungingerV} one can find
that there is a close relation between these questions and a famous
conjecture by De Giorgi. Indeed in \cite{JungingerV,FarinaV} one can
find counterexamples to the PDE version in high enough dimension.

It would be interesting to study these questions in the setting
considered in the present paper.

We recall that De Giorgi conjecture asks whether solutions
$u(x_1,\ldots,x_d)$ of $\triangle u=u-u^3$ which are monotone in
$x_d$, are indeed functions of $\omega\cdot x$, for some
$\omega\in\mathbb{R}^d$ (at least in dimension $d\leq8$).

In Aubry-Mather theory, one considers periodic potentials and one
allows instead of $\triangle$ an elliptic operator with periodic
potentials.

The fact that minimizers are functions of $\omega\cdot x$ is quite
analogous to the fact that they are given by a hull function.
\end{Remark}

\begin{Remark}
The approach based on the Lagrangian \eqref{Percivalian} has been
shown to be a very effective numerical tool \cite{Percival'79}. It
has also been used as the basis of a KAM theory
\cite{LM'01,Rafael'08,MR2507323,SuL11}.
\end{Remark}

In Section~\ref{background}, we will recall the standard definitions
in the calculus of variations adapted to our situation. In
Section~\ref{assumptions}, we will
 detail the assumptions of
our Theorems which we will state and prove in Section 3 (existence
of minimizers), Section 4 (minimizers are ground states) and Section
5 (existence of other critical points).

\subsection{Standing assumptions on the $H_j$}
\label{assumptions}

In order to implement the above program, we will use several
assumptions on the variational principle.

$H_j:\mathbb{R}^2\rightarrow\mathbb{R}$ satisfies periodic condition
(H1) and negative twist condition (H2):
\begin{description}
\item [(H1)] $H_{j}(u+1,v+1)=H_{j}(u,v) \qquad\forall~u,~v\in\mathbb{R},~j=1,\ldots,d$
\item [(H2)] $H_j\in C^2$ and $\partial_{1}\partial_{2}H_{j}\leq c<0,~~ j=1,2,\ldots,d$.
\item [(H3)] $H_j,~j=1,\ldots,d$ have a lower bound.
\end{description}

These assumptions are very representative of the assumptions
customary in Aubry-Mather theory, even if  they can be weakened
slightly.

As a consequence of (H1), we see that the functional
$\mathscr{P}_\omega$ has the following symmetries:
\begin{itemize}
\item (a)~~$\mathscr{P}_\omega(h)=\mathscr{P}_\omega(h+1)$;
\item (b)~~$\mathscr{P}_\omega(h)=\mathscr{P}_\omega(h\circ T_a)$ where $T_a(x)=x+a$.
\end{itemize}
Note that these symmetries make the variational problem
``degenerate". As often used in the calculus of variations one can
overcome this degeneracy by formulating the problem in appropriate
quotient spaces (see \cite{Palais'79}) for a discussion of these
questions.

\begin{Definition}
We call $\omega$ non-resonant if
$\omega_1,\omega_2,\ldots,\omega_d,1$ are rationally independent and
resonant otherwise.
\end{Definition}

\section{Preliminaries}
\label{background}

In this section, we collect some standard definitions from the
calculus of variations that we will use. This section contains only
standard definitions and elementary results and should be used only
as reference.

\subsection{Basic definitions in classical
calculus of variations} We start by summarizing the main concepts in
the sequences approach. This is not the basis of our approach, but
eventually, we will show that the solutions obtained by the hull
function approach lead to sequences which are minimizers in the
sense of calculus of variations.

According to \cite{Morse'24},
\begin{Definition}
A configuration $u: \mathbb{Z}^d \rightarrow \mathbb{R} $
 is
called a class-A minimizer for \eqref{lagrangian} when  for every
$\varphi: \mathbb{Z}^d \rightarrow \mathbb{R}$ with $\varphi_i = 0$
when $|i| \ge N$, we have
\begin{equation}\label{classa}
\sum_{i\in\mathbb{Z}^d, |i| \le N+1 }\sum_{j=1}^d
    H_{j}(u_{i}+ \varphi_i ,u_{i+e_j} + \varphi_{i + e_j} )
\ge \sum_{i\in\mathbb{Z}^d, |i| \le N+1 }\sum_{j=1}^d
    H_{j}(u_{i},u_{i+e_j} )
\end{equation}
\end{Definition}

The equation \eqref{classa} can be interpreted heuristically as
saying $\mathscr{L}(u + \varphi) \ge \mathscr{L}(u )$ after we
cancel the terms on both sides that are identical.

Class-A minimizers are also called \emph{ground states} in the
mathematical physics literature and  \emph{local minimizers} in the
calculus of variations literature.

\begin{Definition}
We say that a configuration is a critical point of the action
whenever it satisfies the Euler-Lagrange equations for every $i \in
\mathbb{Z}^d$
\begin{equation}\label{EulerLagrange}
\sum_j \partial_1 H_j( u_i, u_{i + e_j})  +
\partial_2 H_j( u_{i - e_j}, u_i)  = 0
\end{equation}
\end{Definition}

The equations \eqref{EulerLagrange} are heuristically
$\partial_{u_i} \mathscr{L}(u) = 0$. Note that, even if the sum in
\eqref{lagrangian} is purely formal, the system of  equations
\eqref{EulerLagrange} is well defined. For every $i \in
\mathbb{Z}^d$, equation \eqref{EulerLagrange} involves only a finite
sum of terms.

By considering $\varphi_i = \varepsilon \delta_{i,j}$ where
$\delta_{i,j}$ is the Kronecker delta, it is easy to see that if $u$
is a ground state, then, it satisfies the Euler-Lagrange equations
\eqref{EulerLagrange}. The converse is certainly not true.

\subsection{Order properties of configurations}

Order properties of configurations play a very important role in
Aubry-Mather theory.

The following is a standard definition.
\begin{Definition}\label{Birkhoffconfiguration}
We say that $u:\mathbb{Z}^d\rightarrow \mathbb{R}$ is a Birkhoff
configuration if for every $k\in\mathbb{Z}^d,~~l\in\mathbb{Z}$, we
have either
 \[u_{i+k}+l\geq u_i\qquad \forall~ i\in\mathbb{Z}^d \]
or
\[u_{i+k}+l\leq u_i\qquad \forall~ i\in\mathbb{Z}^d \]
\end{Definition}

In other words, the graph of $u$ does not intersect its horizontal
or vertical translations by integer vectors.

We also have
\begin{equation}\label{rotation}
\omega_j  = \lim_{n \to \pm \infty}\frac{1}{n}  (u_{i + n e_j} -
u_{i})
\end{equation}
and the limit is reached uniformly in $i$.

The notion of Birkhoff configurations was introduced in
\cite{Mather'82}. The name Birkhoff configurations appeared in
\cite{Katok}. These configurations are also called self-conforming
or non-self-intersecting. Their relevance to classical problems in
calculus of variation was emphasized in \cite{Moser'86}.

Birkhoff order properties are closely related to hull functions.

We note that if $h$ is monotone, $h(\theta+1)=h(\theta)+1$ and
$\omega\in \mathbb{R}^d$, then
\[
u_i=h(i\cdot \omega).
\]
Then
\[
u_{i+k}+l=h(i\cdot \omega+k\cdot\omega)+l=h(i\cdot \omega+k\cdot
\omega+l)
\]
and if $k\cdot \omega+l\geq0$, then
$h(i\cdot\omega+k\cdot\omega+l)\geq h(i\cdot\omega)=u_i$.

Therefore, configurations given by hull functions satisfy the
following
\begin{Definition}
 Let $\omega\in \mathbb{R}^d$. We say that $u:\mathbb{Z}^d\rightarrow\mathbb{R} $ is $\omega$-Birkhoff
if
\[\omega\cdot k+l\geq 0, ~~k\in\mathbb{Z}^d,~~l\in\mathbb{Z}\]
implies
\[
u_{i+k}+l\geq u_i\qquad \forall~ i\in\mathbb{Z}^d.
\]
Equivalently,
\[\omega\cdot k+l\leq 0, ~~k\in\mathbb{Z}^d,~~l\in\mathbb{Z}\]
implies
\[
u_{i+k}+l\leq u_i\qquad \forall~ i\in\mathbb{Z}^d.
\]
\end{Definition}

Clearly, $\omega$-Birkhoff configurations are Birkhoff. That is why
\cite{Rafael'07a, Rafael'10} formulated existence and multiplicity
results for $\omega$-Birkhoff orbits.

The converse is close to being true, but it is not exactly true.

First, we note that, given $u$ Birkhoff, there is one and only one
candidate for $\omega$ which would make it $\omega$-Birkhoff
(analogue of rotation number). If this $\omega$ turns out to be
irrationally related, (i.e., $\omega\cdot k+l=0,
~~k\in\mathbb{Z}^d,~~l\in\mathbb{Z}\Longrightarrow k=0,~l=0$) then,
$u$ is $\omega$-Birkhoff. If $\omega$ has some relations, in Remark
\ref{notomegabirkhoff} we will present examples of Birkhoff orbits
with $\omega$ rotation vector which are not $\omega$-Birkhoff.

\begin{Proposition}\label{rotnum}
 Assume $u$ is Birkhoff. Then, there exists $\omega\in\mathbb{R}^d$ such that
\begin{equation*}
 \lim_{n\rightarrow \infty}\frac{1}{n}[u_{i+k\cdot n}-u_i]=\omega\cdot k.
\end{equation*}
 Furthermore,
\[
 |u_{i+k\cdot n}-u_i-n\omega\cdot k|\leq 2
\]
and, if $\omega\cdot k+l>0$ (resp. $\omega\cdot k+l<0$) for
$k\in\mathbb{Z}^d,~~l\in\mathbb{Z}$, we have
\[
 u_{i+k}+l>u_i \qquad \forall~i\in\mathbb{Z}.
\]
(resp. $u_{i+k}+l<u_i \qquad \forall~i\in\mathbb{Z}$)
\end{Proposition}

We note that given a Birkhoff configuration, the sets
\[
 A_\geq=\{(k,l)\in\mathbb{Z}^d\times\mathbb{Z}~|~u_{i+k}+l\geq u_i\},
\]
\[
 A_\leq =\{(k,l)\in\mathbb{Z}^d\times\mathbb{Z}~|~u_{i+k}+l\leq u_i\}
\]
and
\[
 A_= =\{(k,l)\in\mathbb{Z}^d\times\mathbb{Z}~|~u_{i+k}+l= u_i\}
\]
are respectively cones and subspaces. If $(k_1,l_1),~(k_2,l_2)\in
A_\geq$, then $\forall i\in \mathbb{Z}^d$
\[
 u_{i+(k_1+k_2)}+l_1+l_2\geq u_{i+k_1}+l_1\geq u_i
\]
Since $u$ is Birkhoff $A_\geq\cup
A_\leq=\mathbb{Z}^d\times\mathbb{Z}$ and $A_\geq\cap A_\leq=A_=$ (it
could be open).

Proceeding as in the theory of Dedekind cuts, we can find a unique
$\omega\in\mathbb{R}^d$ such that
\begin{align}
 &A_\geq =\{\omega\cdot k+l\geq 0\}\nonumber\\
 &A_\leq =\{\omega\cdot k+l\leq 0\}\nonumber\\
 &A_= =\{\omega\cdot k+l= 0\}.\nonumber
\end{align}
The proof of the existence of the limit can be done exactly as in
the proof of the rotation number in \cite{Poincare}. (See
\cite{Kra'96} for a proof in the context of commuting
diffeomorphisms or \cite{Rafael'98}.)

If there exists $i\in\mathbb{Z}^d$ such that
\[
 u_{i+k}+l\geq u_i.
\]
Because $u$ is Birkhoff, we should have the inequality of all $i$
having
\[
 u_{i+n\cdot k} +n\cdot l\geq u_i
\]
and, taking limits $\omega\cdot k+l\leq 0$. Similarly, we have that
if $u_{i+k}+l\leq u_i, ~~\omega\cdot k+l\geq 0$.

Therefore, we see, comparing with $u_0$ that
\[
 |u_i-u_0-\omega\cdot k|\leq 1.
\]
This, of course establishes that the limit defining the rotation
number is reached uniformly.

\begin{Remark}
We note that Proposition~\ref{rotnum} uses essentially the fact that
we are considering configurations on $\mathbb{Z}^d$, which is a
commutative group. If we consider configurations in non-commutative
groups, it is not clear that for configurations satisfying
Definition~\ref{Birkhoffconfiguration}, we have that the limit
\eqref{rotation} exists and has good properties. Therefore in
\cite{Rafael'07a, Rafael'10}, the Birkhoff orbits are defined as
those which satisfy the conclusions of Proposition~\ref{rotnum}. In
our context, both are equivalent. There are definitions of
$\omega$-Birkhoff orbits for more general latices in Appendix A.
\end{Remark}

\begin{Remark}\label{notomegabirkhoff}
In view of Proposition \ref{rotnum}, the main difference between
$\omega$-Birkhoff and Birkhoff is that when $\omega\cdot k+l=0$,
Birkhoff only claims that we can compare $u_{i+k}+l$ and $u_i$ with
the same sign. The $\omega$-Birkhoff claims that since we have both
inequalities $u_{i+k}+l=u_i$.

This shows how to construct solutions which are Birkhoff with
rotation vector $\omega$ but not $\omega$-Birkhoff. For example in
$d=1$, let $f$ be an orientation preserving diffeomorphism of the
circle with rotation number $=\frac{1}{2}$ with isolated periodic
points of period 2. Any orbit is a Birkhoff sequence, but only the
periodic orbits are $\frac{1}{2}$-Birkhoff.
\end{Remark}

\begin{Proposition}
Assume that $u_i$ is an $\omega$-Birkhoff configuration. Then, there
exists $h:\mathbb{R}\rightarrow\mathbb{R},~~h(\theta+1)=h(\theta)+1$
monotone such that \eqref{hullform} holds
\end{Proposition}
\begin{proof}
The definition of $\omega$-Birkhoff shows that
\[
 u_i+l-u_0
\]
satisfies the same order relation $\omega\cdot i+l$.

Therefore, if we write $u_i+l-u_0$ as a function of $\omega\cdot
i+l$, we will obtain a monotone function $h$ defined on the set
$\{\omega\cdot i+l\}_{i\in \mathbb{Z}^d,~l\in\mathbb{Z}}$. It can be
extended to a monotone function on $[0,1]$.

We also note that because of the way that $l$ enters, we obtain
$h(\theta+1)=h(\theta)+1$. Hence, we can extend the function $h$ to
a hull function.
\end{proof}

\subsection{Spaces for hull  functions, topology and order}

As we indicated, we will present two proofs of Theorem 1. The main
trade-off is between establishing the validity of the Euler-Lagrange
equations and  establishing properties of the minimizers. If we
include spaces of functions that incorporate many properties, then
these properties are, of course, true for the spaces, but, then, it
is hard to establish the Euler-Lagrange equations because we may be
at the boundary of the spaces.

We will start by indicating two different spaces.

\subsubsection{Two spaces of hull functions}\label{definition of
Y^*}

We define the space of functions
\begin{equation}
Y=\{~h~|~h~\text{monotone},~h(\theta+1)=h(\theta)+1,~h(\theta_-)=h(\theta)\}
\end{equation}
This is the space of functions which are monotone -- and therefore
have at most countably many points of discontinuity -- we assume
that the functions are continuous on the left.

Now, we turn to give $Y$ a topology and collect some of the
properties.

We first define
\[{\text{graph}}(h)=\{(\theta,y)\in\mathbb{R}^2:h(\theta)\leq y\leq
h(\theta_+)\}.\] If $h, \tilde{h} \in Y$ we define the distance as
the Hausdorff distance of the graphs.
\begin{equation} \label{graphdistance}
d(h,\tilde{h})=\max \{\sup_{ \xi \in {\text{graph}}(h)} \rho( \xi,
{\text{graph}}(\tilde{h})),
  \sup_{\eta \in {\text{graph}}(\tilde h)}\rho(\eta,{\text{graph}}(h))\}
\end{equation}
 where
 $\rho(\cdot,\cdot)$ is the Euclidean distance from a point to a
set, $\rho(x, S) = \inf_{y \in S} | y -x| $. Note that the graph
topology is weaker than the $L^\infty$ topology.

It is a standard result that the functions $h \in Y$ can be
identified with non-negative periodic  Borel probability measures
times the reals by $h(x) = \mu( [0, x]) + h(0)$. The topology
induced by the distance in \eqref{graphdistance} is the same as the
topology induced by the weak-* convergence in the unit interval. In
dynamics, the measures associated to $h$'s that satisfy the
Euler-Lagrange equation (4) are called Mather measures and are the
basic objects for extending Aubry-Mather theory to higher
codimension in \cite{Mather'89,Mather'91,Mane'96,Contreras'99}.

It is a standard result that
 $Y/\mathbb{R}=Y/\!\!\sim$ is compact
 where $\sim$ is the equivalence relation
defined by $h\!\sim\!\tilde{h}\Leftrightarrow\exists~a$, such that
$\tilde{h}=h\circ T_a$ where $T_a(\theta)=\theta+a$ is a translation
function for all $ \theta,~a\in\mathbb{R}$. Indeed, $Y/\!\!\sim$ is
isomorphic to probability measures on the circle endowed with the
weak-* topology  times the circle. The first factor is compact
because of Banach-Alaoglu theorem and Riesz representation theorem.

Another space that we will consider is $Y^*_N = \{ ~~h \in
L^\infty_{\rm loc}~~~ | ~~~h (\theta +N) = h(\theta) + N\}$ for any
$N\in\mathbb{Z}$.

We consider it endowed with the topology of pointwise convergence.
By Tikhonov theorem, subsets of $Y^*\equiv Y_1^*$ which are bounded
in $||\cdot||_{L^\infty} $ are precompact.

Compared with $Y$, the space $Y^*$ is more flexible because it does
not have the constraint of monotonicity.

\subsubsection{Order properties}

Also, we endow $L^\infty\supseteq Y$ with a partial order given by
$h<\tilde{h}\Leftrightarrow h(\theta)\leq \tilde{h}(\theta)$ for all
$\theta\in\mathbb{R}$ and $h\nequiv\tilde{h}$. We write
$h\prec\!\!\prec\tilde{h}$ to denote $h(\theta) < \tilde{h}(\theta)$
for all $\theta\in \mathbb{R}$.

A small corollary is that, given two functions $h_- \leq h_+$,
$\{~~h \in Y \,~~~ | h_- \leq h \leq h_+ \}$ is compact with the
graph topology. It is clear that it is a closed set of a compact
set.

The analogous set in $Y^*$ $\{~~h \in Y^* \, |~~~ h_- \leq h \leq
h_+ \}$ is also compact for the pointwise convergence topology.

\subsubsection{Some background in lattice theory}

\begin{Definition}[Lattice]
A lattice is a partially ordered set any two of whose elements have
a greatest lower bound and a least upper bound.
\end{Definition}
\begin{Definition}[Complete lattice]
A lattice $\Lambda$ is complete if each $X\subseteq\Lambda$ has a
least upper bound and a greatest lower bound in $\Lambda$.
\end{Definition}
The set $Y_n^*$ has a natural lattice structure induced by the
canonical lattice operations on the real line, i.e.
\[h\vee\tilde{h}(\theta)=\max\{h(\theta),\tilde{h}(\theta)\},\qquad h\wedge\tilde{h}(\theta)=\min\{h(\theta),\tilde{h}(\theta)\}.\]
where $h,\tilde{h}\in Y_n^*$. It is easily seen that
$h\vee\tilde{h},h\wedge\tilde{h}\in Y_n^*$. $Y_1^*$ is not complete
because it includes an $\mathbb{R}$ factor but the subspace
$Y/\!\!\sim$ is.

\begin{Remark}\label{Definitionofa}
The space $Y/\!\!\sim$ is the basis of \cite{Mather'82b}. The paper
\cite{Mather'82} uses the space
\begin{equation*}
X=\{~h\in Y~|~h(\theta)\geq 0 ~for~ \theta>0 ~and~ h(\theta)\leq0
~for~ \theta\leq 0\}.
\end{equation*}
The space X is also compact as shown in \cite{Mather'82}.

The idea is somewhat similar. The reason why Y is not compact is
because it contains an $\mathbb{R}$ factor. (The Borel measure
factor is compact by Banach-Alaoglu theorem.)

Because of the symmetries (a),(b) of the variational principle, we
can formulate the variational problem on ``normalized'' $h$'s. If we
use (a) to normalize $h$ by adding integers we are led to
$Y/\!\!\sim$. If we use (b) to normalize the $h$ by composing with
$a$ translation $T_a$ where $a=\inf\{x~~|~~h(x)\geq0\}$, we are led
to X.
\end{Remark}

In a complete lattice $\Lambda$, we define the order-converge of any
net $\{h_\alpha\}\subseteq\Lambda$. We say that $h_\alpha$ order
converges when
\[\liminf\{h_\alpha\}=\limsup\{h_\alpha\}
\]
where $\liminf\{h_\alpha\}\equiv\sup_\beta\{\inf_{\alpha\geq\beta}
h_\alpha\}$ and
$\limsup\{h_\alpha\}\equiv\inf_\beta\{\sup_{\alpha\geq\beta}
h_\alpha\}$.
\begin{Definition}
A real-valued function $\mathscr{P}$ on a complete lattice $\Lambda$
is called lower semi-continuous if
\[\mathscr{P}(\lim_{j\rightarrow\infty}
h_j)\leq\liminf_{j\rightarrow\infty}\mathscr{P}(h_j)\]
\end{Definition}
whenever the limit exists in $\Lambda$ with respect to the
order-convergence.
\begin{Definition}[Sub-modular]
$\mathscr{P}$ is called sub-modular if for all $h,~\tilde{h}\in
\Lambda$ it satisfies the following inequality:
\[\mathscr{P}(h\vee\tilde{h})+\mathscr{P}(h\wedge\tilde{h})\leq\mathscr{P}(h)+\mathscr{P}(\tilde{h})
\]
where $\vee$ and $\wedge$ are the abstract lattice operations.
\end{Definition}
For example Percival's Lagrangian $\mathscr{P}_\omega$ is lower
semi-continuous and sub-modular (see Lemma~\ref{lemma:Fundamental
Inequality}).

\section{Existence of minimizers and their properties}

In this section, we construct minimizers of $\mathscr{P}_\omega$ in
\eqref{Percivalian}  and show that they are solutions of
Euler-Lagrange equation \eqref{EulerLagrangePercival}.

We present two different functional approaches. One is based on $Y$,
the space of monotone functions, and another one is based on $Y^*_1$
the space of measurable functions and we will show that they
coincide. Later, in Section~\ref{minimizers-ground_states} we will
show that the configurations  generated by $h$ according to
\eqref{hullform} are indeed ground states.

\subsection{A treatment of minimizers based on compactness}

\newtheorem{Theorem}{Theorem}
\begin{Theorem}\label{Theorem:Basic}
Given$~\omega\in\mathbb{R}^d$, the Percival Lagrangian
$\mathscr{P}_\omega$ reaches a minimum  in Y.

Any minimizer satisfies the Euler-Lagrange equation
\eqref{EulerLagrangePercival}.
\end{Theorem}

\begin{proof}
From (H1), the definition of $\mathscr{P}_\omega(h)$, and
$h(\theta+1)=h(\theta)+1$, it follows that $\mathscr{P}_\omega(h)$
is translation invariant (a), (b).

To prove the existence of the minimizer, it suffices to prove the
continuity of $\mathscr{P}_{\omega}$ on $Y$. If it is true, we can
obtain a minimal point on the compact subset $\mathcal{C} = \{ 0 \le
h(\theta) \le 2 \}$.

This minimizer will also be a minimizer in $Y$ because, given $h \in
Y$, we can find $a \in \mathbb{R}, n \in \mathbb{Z}$ such that  $h
\circ T_a + n\in \mathcal{C}$ and
using (a),(b), 
$ \mathscr{P}_\omega( h \circ T_a + n) = \mathscr{P}_\omega(h)$.

In fact, let
  \[M=\max\{1,\max_{j}\sup_{|x-x'|\leq 2}{|\partial_{1}H_{j}(x,x')|},\max_{j}\sup_{|x-x'|\leq 2}{|\partial_{2}H_{j}(x,x')|}\}.\]
Since $H_{j}(u+1,v+1)=H_{j}(u,v)~~\forall~ u,v\in\mathbb{R}$, we
will get $\partial_{1}H_{j}(u+1,v+1)=\partial_{1}
H_{j}(u,v)~,~\partial_{2}H_{j}(u+1,v+1)=\partial_{2} H_{j}(u,v)$. It
follows that $M\leq\infty$. From the definition of
$\mathscr{P}_\omega$ and the mean value theorem, it follows that
\begin{equation}\label{estimate}
|\mathscr{P}_\omega(h)-\mathscr{P}_\omega(\tilde{h})|
               \leq\int_0^1[dM|h(\theta)-\tilde{h}(\theta)|+M\sum_{j=1}^d|h(\theta+\omega_j)-\tilde{h}(\theta+\omega_j)|]d\theta.
\end{equation}

Let $0<\epsilon\leq 1$. Let
$\delta=\delta(\epsilon)=\frac{\epsilon^2}{1000(dM)^2}$. Suppose
$d(h,\tilde{h})<\delta<\frac{1}{1000}$, i.e. for any
$\theta\in\mathbb{R}$, there exists $(\tilde{\theta},\tilde{y})\in
\text{graph}(\tilde{h})$ such that
\[
|(\theta,h(\theta))-(\tilde{\theta},\tilde{y})|<\delta,
\] which implies
\[
|h(\theta)-\tilde{h}(\theta)|<1+\delta<2.
\]

Suppose $a\in \mathbb{R}$. Let
$\pi_a=\{~\theta\in(a,a+1)~|~|h(\theta)-\tilde{h}(\theta)|\geq
\frac{\epsilon}{5dM}~\}$. From the assumption that
$d(h,\tilde{h})<\delta$, i.e. for any $\theta\in\mathbb{R}$, there
exists $(\tilde{\theta},\tilde{y})\in \text{graph}(h)$ such that
\begin{equation*}
\left\{\!\!\!\!
\begin{array}{rl}
& |\theta-\tilde{\theta}|<\delta\\
& |\tilde{h}-\tilde{y}|<\delta
\end{array} \right.
\end{equation*}
we obtain
\begin{equation}\label{estimate1}
h(\theta+\delta)\geq \tilde{h}(\theta)-\delta\geq
h(\theta)+\frac{\epsilon}{5dM}-\delta\geq
h(\theta)+\frac{199\epsilon}{1000dM}
\end{equation}
in the case $\tilde{h}(\theta)\geq h(\theta)+\frac{\epsilon}{5dM}$
and we obtain similarly
\begin{equation}\label{estimate2}
h(\theta-\delta)\leq \tilde{h}(\theta)-\delta\geq
h(\theta)-\frac{\epsilon}{5dM}+\delta\leq
h(\theta)-\frac{199\epsilon}{1000dM}
\end{equation}
in the case $\tilde{h}(\theta)\leq h(\theta)-\frac{\epsilon}{5dM}$.

Let $\pi_a^\prime$ (resp. $\pi_a^{\prime\prime}$) denote the set of
$\theta\in (a,a+1)$ where \eqref{estimate1} (resp.
\eqref{estimate2}) holds. Then
\[
\pi_a\subseteq \pi_a^\prime\cup\pi_a^{\prime\prime}.
\]

At any point $\theta\in\pi_a^\prime$ the variation of $h$ over the
interval $[\theta,\theta+\delta]$ is $\geq
\frac{\epsilon}{5dM}-\delta$. Since the total variation of $h$ over
$(a,a+1)$ is $\leq1$, it follows that $\pi_a^\prime$ can be covered
by at most $[\frac{1000dM}{199\epsilon}]+1\leq 7\frac{dM}{\epsilon}$
intervals of length $\delta$. Hence the measure of $\pi_a^\prime$ is
at most $\frac{7dM\delta}{\epsilon}\leq\frac{\epsilon}{100dM}$.
Similarly, the measure of $\pi_a^{\prime\prime}$ is $\leq
\frac{\epsilon}{100dM}$. Hence the measure of $\pi_a$ is $\leq
\frac{\epsilon}{50dM}$.

Since $|h(\theta)-\tilde{h}(\theta)|\leq 2$ for all
$\theta\in\mathbb{R}$ and
$|h(\theta)-\tilde{h}(\theta)|\leq\frac{\epsilon}{5dM}$ for
$\theta\in (0,1)-\pi_0$ and for
$\theta\in(\omega_j,\omega_j)-\pi_{\omega_j}$, we obtain from
\eqref{estimate} that
\[
|\mathscr{P}_\omega(h)-\mathscr{P}_\omega(\tilde{h})|\leq
dM\cdot(\frac{4\epsilon}{50dM}+\frac{\epsilon}{5dM})<\epsilon.
\]
This completes the proof of the existence of the minimizer.

Now we go into the proof that the minimizer satisfies the
Euler-Lagrange equation. The proof below is similar to
\cite{Mather'82}. The key point in the proof is that given a
minimizer we can find enough deformations that do not leave the
space so that we can conclude that the Euler-Lagrange equations
hold. These arguments are sometimes called in the calculus of
variations \emph{deformation lemmas}, a name which is used with
another meaning in other fields.

\begin{Lemma}
Suppose $a\leq0\leq b$ and $a<b$. Suppose an element $h_s$ of $Y$ is
given for $a\leq s\leq b,~h_s(\theta)$ is $C^2$ function of s for
each fixed $\theta$, and $\frac{\partial }{\partial
s}h_s(\theta),~~\frac{\partial^2 }{\partial s^2}h_s(\theta)$ are
uniformly bounded and measurable for $a\leq s\leq
b,~\theta\in\mathbb{R}$. Then
\begin{equation}\label{Variation}
\frac{d}{ds}\mathscr{P}_\omega(h_{s})|_{s=0}=\int_0^1X(h)\cdot\dot
h(\theta)d\theta,
\end{equation}
where $\dot h_s(\theta)=\frac{\partial }{\partial
s}h_s(\theta),~\dot h(\theta)=\dot h_0(\theta)$ and $h=h_0$.
\end{Lemma}
Clearly, if $h$ is a minimizer and $h_s$ is a deformation, $h_0=h$,
we have $\frac{d}{ds}\mathscr{P}_\omega(h_s)|_{s=0}=0$. Using
\eqref{Variation} we obtain that $\int_0^1 X(h)\cdot\dot h=0$. To
conclude that $X(h)$ is identically zero, we have to argue that we
can obtain enough deformations $\dot h(\theta)$ that force that
$X(h)$ is zero in the neighborhood of any point
$\theta\in\mathbb{T}^1$. We will generate deformations by solving
the ordinary differential equation:
\begin{equation*}
\left\{\!\!\!\!
  \begin{array}{rl}
   &\frac{d}{ds}u_s(\theta)=\rho\circ\pi\circ u_s(\theta)\\
   &u_0=id,
  \end{array}
\right.
\end{equation*}
where $\pi:\mathbb{R}\rightarrow\mathbb{R}/\mathbb{Z}$ is the
projection map and $\rho$ which has values in $[0,1]$ will be
decided later. We will consider every continuous point $\theta_0$ of
$h$ first and then take the limit to approximate the discontinuous
ones due to the fact that $h$ is monotone. We simplify the formula
in the above lemma in the two cases below:
\begin{enumerate}
\item When $h^{-1}\circ h(\theta_0)$ is a single point,
we define $h_s=u_s\circ h$ and get
\[
\frac{d}{ds}\mathscr{P}_\omega(h_{s})|_{s=0}=\int_0^1X(h)\cdot\rho\circ\pi\circ
h ~d\theta.\]
\item When $h^{-1}\circ h(\theta_0)$ is an interval, there exists
$\theta_1>\theta_0$ such that
\begin{itemize}
\item we define
\begin{equation*}
\psi_s(\theta)=\left\{
\begin{array}{rl}
u_s\circ h(\theta) & \text{if } \exists~n\in\mathbb{Z}\text{ such
that
}\theta_0+n<\theta\leq\theta_1+n\\
h(\theta)          & \text{otherwise }
\end{array} \right.
\end{equation*}
and get
$\frac{d}{ds}\mathscr{P}_\omega(\psi_{s})|_{s=0}=\int_{\theta_0}^{\theta_1}X(h)\cdot\rho\circ\pi\circ
h ~d\theta$.
\item
\begin{equation*}
\xi_s(\theta)=\left\{
\begin{array}{rl}
h(\theta)   & \text{if } \exists~n\in\mathbb{Z}\text{ such that
}\theta_0+n<\theta\leq\theta_1+n\\
u_s\circ h(\theta)        & \text{otherwise }
\end{array} \right.
\end{equation*}
and get
$\frac{d}{ds}\mathscr{P}_\omega(\xi_{s})|_{s=0}=\int_{\theta_1-1}^{\theta_0}X(h)\cdot\rho\circ\pi\circ
h~ d\theta$.
\end{itemize}
\end{enumerate}

For case (1), provided that $\rho$ has support in a sufficiently
small neighborhood of $\pi\circ h(\theta_0)$, we have $h_s\in Y$ for
$s$ sufficiently small.

The hypothesis that $\mathscr{P}_\omega$ takes its minimum at
$h=h_0$ implies $\frac{d}{ds}\mathscr{P}_\omega(h_{s})|_{s=0}=0$.
Since $X(h)$ is continuous at $\theta=\theta_0$, and
$\theta_0=h^{-1}h(\theta)$, the fact that $\int_0^1
X(h)\cdot\rho\circ\pi\circ h ~d\theta=0$ for all $\rho$ of the type
we consider, implies $X^{\theta_0}(h)=0$. (Here we write explicitly
the dependence of $X(h)$ on the point $\theta_0$ we choose.)

For case (2), let $\alpha$ and $\beta$ be endpoints of
$h^{-1}h(\theta_0)$ with $\alpha<\beta$. $X(h)$ is a decreasing
function of $\theta\in(\alpha,\beta)$ by (H2). It is easy to see
that if $\rho$ has support in a sufficiently small neighborhood of
$\pi\circ h(\theta)$, then $\psi_s\in Y$ for $s\geq0$ sufficiently
small and $\xi_s\in Y$ for $s\leq0$ sufficiently small. The
assumption that $\mathscr{P}_\omega$ takes its minimum at
$h=\psi_0=\xi_0$ implies
$\frac{d}{ds}\mathscr{P}_\omega(\psi_{s})|_{s=0}\geq0$ and
$\frac{d}{ds}\mathscr{P}_\omega(\xi_{s})|_{s=0}\leq0$. In view of
the fact that $X(h)$ is a decreasing function on $(\alpha,\beta)$,
we have $X(h)\geq0$ and $X(h)\leq0$. Hence $X^{\theta_0}(h)=0$. This
completes the proof of the second part of Theorem
\ref{Theorem:Basic}.
\end{proof}

\subsection{Existence of minimizers based on order properties}

In the following we present another approach to the same problem
based on different spaces. Basically, we show that
$\mathscr{P}_\omega$ reaches a minimum  on $Y^*$.

\begin{Theorem}\label{Th:Extended}
Under our standing assumptions, there is a minimizer of
$\mathscr{P}_\omega$ over $Y^*$. Any minimizer on $Y^*$ satisfies
the Euler-Lagrange equations.

There is one minimizer which lies on $Y$.
\end{Theorem}

Of course, once we prove that there is one minimizer in the whole
space $Y^*$ which actually lies in $Y$ we conclude that $\inf_{h \in
Y^*} \mathscr{P}_\omega(h)   = \inf_{h \in Y} \mathscr{P}_\omega(h)$
and, therefore that all the minima in $Y$ are also minima in $Y^*$.

The main advantage of this argument is that, since $Y^*$ does not
involve any constraints, the deformation lemmas are almost trivial
and, therefore it is easy to show that the minimizers satisfy the
Euler-Lagrange equations \eqref{EulerLagrangePercival}.




\subsection{Proof of Theorem~\ref{Th:Extended}}


We use  the following basic lemma in \cite{Forni'96}:
\begin{Lemma} \label{lemma:Basic}
Let $\mathscr{P}$ be a real-valued function on a complete lattice
$\Lambda$. Suppose $\mathscr{P}$ is sub-modular, lower
semi-continuous and bounded from below. Then $\mathscr{P}$ has a
minimum on $\Lambda$.
\end{Lemma}
\begin{proof}
Since $\mathscr{P}$ is bounded from below,
$\beta=\inf_\Lambda\mathscr{P}$ is a real number. Suppose given any
sequence of positive real numbers $(\epsilon_j)_{j\in\mathbb{N}}$
converging to zero, there exists a sequence
$(h_j)_{j\in\mathbb{N}}\subseteq\Lambda$ such that:
\[\beta\leq\mathscr{P}(h_j)\leq\beta+\epsilon_j
\]
By the sub-modularity, since $\mathscr{P}(h_j\vee
h_{j+1})\geq\beta$, we have
\[\mathscr{P}(h_j\wedge
h_{j+1})\leq\beta+\epsilon_j+\epsilon_{j+1}
\]
and by induction, we have
\[\mathscr{P}(h_j\wedge
h_{j+1}\wedge\cdots\wedge
h_{j+k})\leq\beta+\epsilon_j+\epsilon_{j+1}+\cdots+\epsilon_{j+k}
\]
Define $\tilde{h}_{j,k}\equiv h_j\wedge h_{j+1}\wedge\cdots\wedge
h_{j+k}$ for all $j,k\geq1$. By construction, it is a non-increasing
sequence included in $\Lambda$ with respect to $k$. Due to the
completeness of $\Lambda$, we can define
$\tilde{h}_j=\lim_{k\rightarrow\infty}\tilde{h}_{j,k}\in\Lambda$.
Consequently, one gets:
\begin{equation}\label{above inequality}
\mathscr{P}(\tilde{h}_j)=\mathscr{P}(\lim_{k\rightarrow\infty}\tilde{h}_{j,k})\leq\liminf_{k\rightarrow\infty}\mathscr{P}(\tilde{h}_{j,k})\leq\beta+r_j
\end{equation}
by the above inequality and the lower semi-continuity of
$\mathscr{P}$, where $r_j=\sum_{k\geq j}\epsilon_k$ is finite and
converges to zero, as $j\rightarrow\infty$ if we choose
$(\epsilon_j)_{j\in\mathbb{N}}$ such that
\[\sum_{j\geq0}\epsilon_j<\infty.
\]
On the other hand, since $\tilde{h}_{j,k}\leq\tilde{h}_{j+1,k-1}$
for all $j,k\geq1$, then $(\tilde{h}_j)_{j\in\mathbb{N}}$ is a
non-decreasing sequence, hence it has a limit $\tilde{h}\in\Lambda$.
By the lower semi-continuity and the choice of the sequence
$(\epsilon_j)_{j\in\mathbb{N}}$, \eqref{above inequality} implies:
\[\mathscr{P}(\tilde{h})\leq
\beta+\liminf_{j\rightarrow\infty}r_j=\beta
\]
which concludes the proof, by showing that $\tilde{h}$ is a minimum
point for $\mathscr{P}.\qedhere$
\end{proof}

Now we turn to show how the concrete functional $\mathscr{P}_\omega$
in \eqref{Percivalian} satisfies the assumptions of the abstract
results.

\begin{Lemma}[Fundamental Inequality in Aubry-Mather theory]\label{lemma:Fundamental Inequality}
If $h,~\tilde{h}\in Y_n^*$, then
\begin{equation}\label{fundamentaleq}
\mathscr{P}_\omega(h\vee\tilde{h})+\mathscr{P}_\omega(h\wedge\tilde{h})\leq\mathscr{P}_\omega(h)+\mathscr{P}_\omega(\tilde{h})
\end{equation}
\end{Lemma}
\begin{proof}
Using the fundamental theorem of calculus, we have
\begin{align}
&H_j(h(\theta)\wedge \tilde{h}(\theta),h(\theta+\omega_j)\wedge
\tilde{h}(\theta+\omega_j))+H_j(h(\theta)\vee
\tilde{h}(\theta),h(\theta+\omega_j)\vee
\tilde{h}(\theta+\omega_j))\nonumber\\
&-H_j(h(\theta),h(\theta+\omega))-H_j(\tilde{h}(\theta),\tilde{h}(\theta+\omega))\nonumber\\
=&\qquad\int_{h(\theta)\wedge \tilde{h}(\theta)}^{h(\theta)\vee
\tilde{h}(\theta)}\int_{h(\theta+\omega_j)\wedge
\tilde{h}(\theta+\omega_j)}^{h(\theta+\omega_j)\vee
\tilde{h}(\theta+\omega_j)}\partial_1\partial_2H_j(x,y)dx
dy.\nonumber
\end{align}
Adding over $j$ and integrating with respect to $\theta$, we obtain
\begin{align}
&\mathscr{P}_\omega(h\vee\tilde{h})+\mathscr{P}_\omega(h\wedge\tilde{h})-\mathscr{P}_\omega(h)-\mathscr{P}_\omega(\tilde{h})\nonumber\\
=&\frac{1}{n}\sum_{j=1}^d\int_0^n d\theta\int_{h(\theta)\wedge
\tilde{h}(\theta)}^{h(\theta)\vee
\tilde{h}(\theta)}\int_{h(\theta+\omega_j)\wedge
\tilde{h}(\theta+\omega_j)}^{h(\theta+\omega_j)\vee
\tilde{h}(\theta+\omega_j)}\partial_1\partial_2H_j(x,y)dx dy\leq
0\nonumber
\end{align}
The last inequality holds because of (H2). Hence we obtain
\eqref{fundamentaleq}.
\end{proof}

We had defined before two spaces $Y_1^*$ (see Section
\ref{definition of Y^*}) and $X$ (see Remark \ref{Definitionofa}).
We have $X\subseteq Y_1^*$ (roughly $X$ is a subset of functions in
$Y_1^*$ with some monotonicity properties). Since $X\subseteq Y_1^*$
it is clear that
\[
\inf_{h\in Y_1^*}\mathscr{P}_\omega(h)\leq\inf_{h\in X}
\mathscr{P}_\omega(h).
\]
In Lemma \ref{lemma:Last}, we show that both minima are actually
equal.

This is useful because it is easier to show that minimizers in
$Y_1^*$ satisfy the Euler-Lagrange equation. (Since $Y_1^*$ has less
properties, it is easy to construct deformations that do not leave
the space.)

\begin{Lemma}\label{lemma:Last}
Let $\tilde{h}\in Y_1^*$. Then there exists $\overline{h}\in X$ such
that
\[\mathscr{P}_\omega(\overline{h})\leq\mathscr{P}_\omega(\tilde{h}).
\]
\end{Lemma}
\begin{proof}
Let $Y_A(\tilde{h})\subseteq Y_1^*$ be the complete lattice
generated by the set $\{\tilde{h}\circ T_a: 0\leq a\leq A\}$, which
exists since $\tilde{h}$ is locally bounded, i.e. $Y_A(\tilde{h})$
is the smallest complete lattice that includes the set
$\{\tilde{h}\circ T_a: 0\leq a\leq A\}$. Applying Lemma
\ref{lemma:Basic}, there exists $h_A\in Y_A(\tilde{h})$  for each
$A\geq0$ which minimizes $\mathscr{P}_\omega$ over $Y_A(\tilde{h})$.
Next, we will consider the quotient set $\Lambda_A(\tilde{h})\equiv
Y_A(\tilde{h})/\mathbb{R}$, obtained by projecting the sub-lattices
$Y_A(\tilde{h})$ into the quotient space $Y_1^*/\mathbb{R}$. Since
$\tilde{h}$ is locally bounded and satisfies
$\tilde{h}(\theta+1)=\tilde{h}(\theta)+1$, the sets
$\Lambda_A(\tilde{h})$ stabilize as $A\rightarrow \infty$, i.e.,
there exists $M>0$ such that
\[
\Lambda_A(\tilde{h})=\Lambda_M(\tilde{h})~~\text{if}~A\geq M.
\]
Hence, due to the translation invariance of $\mathscr{P}_\omega$, it
is possible to choose $h\in Y_M(\tilde{h})$ for each $A\geq M$ that
minimizes $\mathscr{P}_\omega$ over $Y_M(\tilde{h})$. From the
translation invariance and sub-modularity property of
$\mathscr{P}_\omega$, we obtain
\[
\mathscr{P}_\omega(h\wedge h\circ T_a)=\mathscr{P}_\omega(h\vee
h\circ T_a)=\mathscr{P}_\omega(h)=\mathscr{P}_\omega(h\circ T_a),
\]
since $h$ minimizes $\mathscr{P}_\omega$ over $Y_A(\tilde{h})$ if
$A$ is sufficiently large. Let $\{a\}_{i\in\mathbb{N}}$ be an
enumeration of all the positive rational numbers and
$\tilde{h}_m=h\wedge h\circ T_{a_1}\wedge\ldots\wedge h\circ
T_{a_m}$. Repeating the argument above $m$ times, we get
\[\mathscr{P}_\omega(\tilde{h}_m)=\mathscr{P}_\omega(h).
\]
We have $h\geq\tilde{h}_1\geq \ldots\geq \tilde{h}_m\geq \ldots$ and
$\tilde{h}_m|_{[a,b]}\geq C=\inf\{h(s):a\leq s\leq a+1\}$ for all
$m$ and each finite interval $[a,b]$ since $h$ is locally bounded
and $h(\theta+1)=h(\theta)+1$. Consequently,
$\tilde{h}_\infty(\theta)=\lim_{m\rightarrow\infty}\tilde{h}_m(\theta)$
exists for all $\theta\in \mathbb{R}$ and by the lower
semi-continuity of $\mathscr{P}_\omega$ we get
\begin{equation}\label{above inequality2}
\mathscr{P}_\omega(\tilde{h}_\infty)\leq\lim_{m\rightarrow\infty}\mathscr{P}_\omega(\tilde{h}_m)=\mathscr{P}_\omega(h).
\end{equation}

It is sufficient to prove that $\tilde{h}_\infty$ is
order-preserving almost everywhere (i.e. that it agrees with an
order-preserving function almost everywhere). We adapt an argument
in \cite[Lemma 7.2]{Mather'85}.

We know that $\tilde{h}_\infty(\theta)=\inf \{h(\theta+a):a$ is a
positive rational number$\}$ is order-preserving except on a set of
zero measure. In fact, for a positive rational number $a$, we have
$\tilde{h}_\infty\circ T_a \geq\tilde{h}_\infty$ by the definition
of $\tilde{h}_\infty$. Since $\tilde{h}_\infty\in
L_{loc}^\infty(\mathbb{R})$, i.e. is measurable and bounded on
bounded sets, we have $\int_I|\tilde{h}_\infty \circ
T_a-\tilde{h}_\infty \circ T_b|\rightarrow0$ as $a\rightarrow b$,
for any finite interval $I$. Therefore, if $b>0$, we obtain
$\tilde{h}_\infty\circ T_b \geq\tilde{h}_\infty$ almost everywhere.
In other words, for each $b>0$,
$Leb\{\theta:~\tilde{h}_\infty(\theta)>\tilde{h}_\infty(\theta+b)\}=0$
where $Leb$ is the Lebesgue measure. We obtain
\[
Leb\{(\theta,s)\in\mathbb{R}^2:~(\theta-s)(\tilde{h}_\infty(\theta)-\tilde{h}_\infty(s))<0\}=0
\]
by Fubini's theorem. So there exists a set $E\subseteq\mathbb{R}$
and $Leb(E)=0$ such that if $\theta\notin E$, we have
$Leb\{s\in\mathbb{R}:~(\theta-s)(\tilde{h}_\infty(\theta)-\tilde{h}_\infty(s))<0\}=0$.
Hence for $\theta,s\notin E,~\theta<s$, and $a.e. ~\theta<u<s$
$\tilde{h}_\infty(\theta)\leq\tilde{h}_\infty(u)$ and
$\tilde{h}_\infty(u)\leq\tilde{h}_\infty(s)$, that is,
$\tilde{h}_\infty(\theta)\leq\tilde{h}_\infty(s)$ holds for $a.e.~
\theta,s\notin E$. (It is easy to see that
$\tilde{h}_\infty(\theta)=ess.\inf_{s\geq \theta}h(s)$ holds $a.e.
~\theta$.)

Take $\overline{h}\in Y$ such that
$\overline{h}(\theta)=\tilde{h}_\infty(\theta)
~a.e.~\theta\in\mathbb{R}$. We have
$\mathscr{P}_\omega(\overline{h})=\mathscr{P}_\omega(\tilde{h}_\infty)\leq\mathscr{P}_\omega(h)$
due to \eqref{above inequality2}. The final step is to choose $a$
such that $h_0=\bar{h}\circ T_a\in X$. This choice is explained at
the end of Remark \ref{Definitionofa}. Then, we have
\[
\mathscr{P}_\omega(h_0)=\mathscr{P}_\omega(\overline{h})\leq\mathscr{P}_\omega(h)\leq\mathscr{P}_\omega(\tilde{h}).
\]
where the first equality is a consequence of the translation
invariance of $\mathscr{P}_\omega$, the next inequality follows from
the property of $\overline{h}$ and the last inequality holds because
$h$ minimizes $\mathscr{P}_\omega$ over $Y_A(\tilde{h})$ for $A$
sufficient large. This completes the proof.
\end{proof}

\section{Minimizers of the Percival Lagrangian
give rise to ground  states. } \label{minimizers-ground_states}

In this section, we prove that the minimizers of
$\mathscr{P}_\omega$ give rise to ground states when $\omega$ is
both non-resonant and resonant.

\begin{Theorem}\label{aregroundstates}
Let $h_\omega$ be a minimizer of $\mathscr{P}_\omega$ as in Theorem
\ref{Theorem:Basic}. The configuration
$u_i=h_\omega(\theta+\omega\cdot i)$ when $\omega$ is non-resonant
is an $\omega$-Birkhoff ground state.
\end{Theorem}
Out of Theorem \ref{aregroundstates} we can obtain several results
using approximation arguments. We present two representative
results, Corollary \ref{Birkhoffgroundstate} (based on approximation
in the orbit formalism) and Corollary \ref{WBirkhoffgroundstate}
(based on the hull function formalism). Since the result of
Corollary \ref{Birkhoffgroundstate} is based on choices of
approximating subsequences, it is not clear that the orbits produced
are the same.
\begin{Corollary}\label{Birkhoffgroundstate}
Given any frequency $\omega\in \mathbb{R}^d$, there is a Birkhoff
ground state of frequency $\omega$.
\end{Corollary}
It is amusing to note that in the orbit based approach
\cite{Rafael'98, Rafael'07a, Rafael'10}, it is more convenient to
construct ground states of non-resonant frequencies approximating
them by ground states of rational frequencies. Now, we find it more
convenient to construct ground states of non-resonant frequencies
and use an approximation argument to get those of rational
frequencies.
\begin{proof}
The proof of the corollary \ref{Birkhoffgroundstate} is very simple.
We observe that given any $\omega$, we can find a sequence
$\omega_n$ of nonresonant vectors such that $\lim_{n \to \infty}
\omega_n = \omega$. Denote by $u^n$, the ground states corresponding
to this sequence. By the invariance of the action under addition of
integers we can assume that that $u^n_0 \in [0, 1)$ and by the
Birkhoff property, $| u^n_i - \omega_n \cdot i| \le 2$. It follows
that, using the diagonal trick, we can assume that $\lim_{n \to
\infty}  u^n_i  = u^*_i$ exists for all $i \in \mathbb{Z}^d$.

Then, it is a classical argument in \cite{Morse'73} to show that
$u^*$ is a ground state. Suppose by contradiction that we could find
$\varphi$ such that $\varphi_i = 0, |i| \ge N-1$ and that $
\mathscr{L}_N(u^*) - \mathscr{L}_N(u^* + \varphi) \ge \delta > 0$.
Since $\mathscr{L}_N$ involves only finitely many sites, we can find
$n^*$ such that $ \mathscr{L}_N(u^{n^*}) - \mathscr{L}_N(u^{n^*} +
\varphi) \ge \delta/2 > 0$. This is a contradiction with $u^{n^*}$
being a ground state.

To finish the argument, we show that the limiting sequence is
Birkhoff. Fixed $k\in\mathbb{Z}^d,~~l\in\mathbb{Z}$, we can find an
infinite sequence of $n$'s in which the comparison between
$(\tau_k\circ R_l) u^n$ (where $\tau_k$ and $R_l$ are the horizontal
and vertical translations respectively) and $u^n$ has the same sign.
Therefore, the limit of $(\tau_k\circ R_l) u^*$ can be compared with
$u^*$.

Of course, it is perfectly possible that for each of the two
possible comparison signs between $(\tau_k\circ R_l) u^n$ and $u^n$,
there are infinitely many $n$'s. In this case, $u^*$ would satisfy
both comparisons.
\end{proof}

\begin{Corollary}\label{WBirkhoffgroundstate}
$u_i=h_\omega(\theta+\omega\cdot i)$ is an $\omega$-Birkhoff ground
state for any rotation vector $\omega\in \mathbb{R}^d$.
\end{Corollary}

The proof of Corollary \ref{WBirkhoffgroundstate} uses the fact that
the hull function we obtain satisfies the non-symmetry breaking
property, which means $\mathscr{P}_\omega$ reach the same minimum
over $Y^*_1$ and $Y^*_n$ for any $n\in \mathbb{Z}$ (see \cite[Lemma
7.3]{Mather'85}). Since the technique of the proof of Corollary
\ref{WBirkhoffgroundstate} is very similar to that of Theorem
\ref{aregroundstates} we postpone it.
\begin{Remark}
In fact, Corollary \ref{WBirkhoffgroundstate} implies Corollary
\ref{Birkhoffgroundstate}. The non-symmetry breaking property plays
an important role here.
\end{Remark}

\subsection{Proof of Theorem~\ref{aregroundstates} }

We use arguments inspired by \cite{Mather'85} but we require some
more detailed computations.

Suppose $u_i$ is not a ground state, so there exists a configuration
$\tilde{u}_i$ and $K\in\mathbb{Z}^+$ such that $\tilde{u}_i=u_i$ if
$|i|\geq~K$ and $\mathscr{L}_K(\tilde{u})<\mathscr{L}_K(u)$. Let
$1\gg\delta >0$ and set
\begin{align}
&\tilde{h}(\theta)=\tilde{u}_i,\qquad if ~~ |i|\leq K,~t+\omega\cdot
i-\delta\leq\theta\leq t+\omega\cdot
i,\nonumber\\
&\tilde{h}(\theta+1)=\tilde{h}(\theta)+1,~~for ~all
~\theta,~~and\nonumber\\
&\tilde{h}(\theta)=h_\omega(\theta),~whenever~\tilde{h}(\theta)~is~not~
defined~by~the~previous~two~conditions.\nonumber
\end{align}
Since $\omega$ is non-resonant and $\delta$ is small, there is no
contradiction between the first two conditions. Consequently,
$\tilde{h}$ is well-defined and $\tilde{h}\in Y_1^*$.
\begin{align}
\mathscr{P}_{\omega}(h_\omega)-\mathscr{P}_{\omega}(\tilde{h})=&\sum_{j=1}^d\int_a^{a+1}[H_j(h_\omega(\theta),h_\omega(\theta+\omega_j))-H_j(\tilde{h}(\theta),\tilde{h}(\theta+\omega_j))]d\theta\nonumber\\
                                                       =&\int_{t-\delta}^{t}[A(\theta)+B(\theta)+C(\theta)]d\theta\nonumber
\end{align}
where
\begin{align}
A(\theta)=&\mathscr{L}_K(h_\omega(\theta+\omega\cdot i))-\mathscr{L}_K(\tilde{u})\nonumber\\
B(\theta)=&\sum_{j=1}^d[\sum_{\substack{i\in\mathbb{Z}^d,|i|=K\\i_j\geq0}}H_j(h_\omega(\theta+\omega\cdot i),h_\omega(\theta+\omega\cdot i+\omega_j))\nonumber\\
                       &\qquad-\sum_{\substack{i\in\mathbb{Z}^d,|i|=K\\i_j\geq0}}H_j(h_\omega(t+\omega\cdot i),h_\omega(\theta+\omega\cdot i+\omega_j))]\nonumber\\
C(\theta)=&\sum_{j=1}^d[\sum_{\substack{i\in\mathbb{Z}^d,|i|=K\\i_j\leq0}}H_j(h_\omega(\theta+\omega\cdot i-\omega_j),h_\omega(\theta+\omega\cdot i))\nonumber\\
                       &\qquad-\sum_{\substack{i\in\mathbb{Z}^d,|i|=K\\i_j\leq0}}H_j(h_\omega(\theta+\omega\cdot i-\omega_j),h_\omega(t+\omega\cdot i))]\nonumber
\end{align}


Clearly, $A(\theta)\rightarrow
A(t)=\mathscr{L}_K(u)-\mathscr{L}_K(\tilde{u})>0$,
$B(\theta)\rightarrow0$ and $C(\theta)\rightarrow0$ as
$\theta\uparrow t$. Consequently,
$\mathscr{P}_{\omega}(h)-\mathscr{P}_{\omega}(\tilde{h})>0$ for
$\delta>0$ small enough. But this contradicts the fact that $h$
minimizes $\mathscr{P}_\omega$ over $Y_1^*$ if we already know the
fact that $h_\omega$ minimizes $\mathscr{P}_\omega$  over
$Y_1^*$.\qed

\begin{proof}[Proof of Corollary \ref{WBirkhoffgroundstate}]
Suppose $u_i$ is not a ground state, so there exists a configuration
$\tilde{u}_i$ and $K\in\mathbb{Z}^+$ such that $\tilde{u}_i=u_i$ if
$|i|\geq~K$ and $\mathscr{L}_K(\tilde{u})<\mathscr{L}_K(u)$. Let
$1\gg\delta >0$ and set
\begin{align}
&\tilde{h}(\theta)=\tilde{u}_i,\qquad if ~~ |i|\leq K,~t+\omega\cdot
i-\delta\leq\theta\leq t+\omega\cdot
i,\nonumber\\
&\tilde{h}(\theta+N)=\tilde{h}(\theta)+N,~~for ~all
~\theta,~~and\nonumber\\
&\tilde{h}(\theta)=h_\omega(\theta),~whenever~\tilde{h}(\theta)~is~not~
defined~by~the~previous~two~conditions.\nonumber
\end{align}
Consequently, $\tilde{h}$ is well-defined and $\tilde{h}\in Y_N^*$
for some sufficiently large $N$. By non-symmetry breaking property,
we get the same contradiction.
\end{proof}

\section{Existence of non-minimal critical points}
We will refer to the functions given in the form \eqref{hullform} as
``quasi-periodic". In some literature, the term quasi-periodic  is
reserved for situations when $h$ is smooth, whereas we will accept
$h$ which are discontinuous. In some literature, these functions are
given the name ``almost-automorphic" in \cite{Ellis}. We will follow
the customary notation in the calculus of variations. One of the
most interesting phenomena in Aubry-Mather theory is that the
quasi-periodic solutions obtained may be discontinuous.

We note that the discontinuity of the minimizers has profound
physical and dynamical interpretations. In the solid state physical
interpretation, if $h\circ T_a$ is a continuous family of critical
points, the physical system can ``slide" whereas if $h\circ T_a$
involves discontinuity, the system is ``pinned". In the case of
twist maps, that $h_\omega$ is continuous corresponds to an
invariant orbit, which is a complete barrier for transport.


In this section, we study the situation when there are two
minimizers which are comparable. Similar results in PDE were studied
in \cite{LlaveVPDE}. There are other more delicate results that show
that if there are gaps in the range of $h$, then there is another
minimizing sequence \cite{Mather'86}. In \cite{LlaveVPDE}, one can
find a proof using the gradient flow approach in spaces of
sequences. We do not present these results here. Indeed we do not
know how do they fit in the hull function approach, except in the
rational frequency case.

\begin{Theorem}\label{Theorem:Main2}
Suppose $h^-< h^+$ are both minimizers of $\mathscr{P}_\omega$ on Y
with frequency vector $\omega$ not completely resonant (not all the
components of $\omega$ are rational numbers). Then,
\begin{enumerate}
\item $h^-\prec\!\!\prec h^+$;
\item There exists a critical point $h^0$ of $\mathscr{P}_\omega$ such
that $h^-\prec\!\!\prec h^0\prec\!\!\prec h^+$ holds;
\end{enumerate}
\end{Theorem}
To prove Theorem \ref{Theorem:Main2}, we will use the gradient flow
method (see \cite{Rafael'97,Gole'01}) for $Y\subseteq L^\infty$.
\begin{Lemma}\label{lemma:gf}
Assume $\partial_1H_j,~\partial_2H_j$ are uniformly $C^r,~r\geq1$.
The infinite system of ODE's:
\begin{equation}\label{ODEsystem}
\left\{\!\!\!
  \begin{array}{rl}
   &\frac{d}{dt} h^t=-X(h^t)\equiv -\sum_{j=1}^d[\partial_1 H_j(h^t,h^t\circ T_{\omega_j})+\partial_2H_j(h^t\circ T_{-\omega_j},h^t)]\\
   &h^0=h_0
  \end{array}
\right.
\end{equation}
defines a $C^r$ flow $\Phi^t$ on $L^\infty$. The rest points of
$\Phi^t$ correspond to critical points of the Percival Lagrangian
$\mathscr{P}_\omega$.
\end{Lemma}
By ODE theory in Banach space (see \cite{Hale'80}), it is easy to
see that the gradient flow $\Phi^t$ is well-defined for $t\geq0$
since the vector field $-X(h^t)$ is globally Lipschitz. From the
gradient flow equation itself, we can get some simple properties.

\begin{Proposition}\label{Proposition:1}
$             [\Phi^t(h_0)]\circ T_a=\Phi^t(h_0\circ T_a);\qquad
              \Phi^t(h_0+m)=\Phi^t(h_0)+m$\\
for any $a\in\mathbb{R},m\in\mathbb{Z}$.
\end{Proposition}

\begin{proof}For the first equality, we differentiate its left hand
side with respect to t, and get:
\begin{align}
\frac{d}{dt}[\Phi^t(h_0)]\circ T_a&=-X(\Phi^t(h_0))\circ T_a\nonumber\\
                                &=-\sum_{j=1}^d[\partial_1 H_j(\Phi^t(h_0),\Phi^t(h_0)\circ T_{\omega_j})+\partial_2H_j(\Phi^t(h_0)\circ T_{-\omega_j},\Phi^t(h_0))]\circ T_a\nonumber\\
                                &=-\sum_{j=1}^d[\partial_1 H_j(\Phi^t(h_0)\circ T_a,\Phi^t(h_0)\circ T_{\omega_j}\circ T_a)+\partial_2H_j(\Phi^t(h_0)\circ T_{-\omega_j}\circ T_a,\Phi^t(h_0)\circ T_a)]\nonumber\\
                                &=-X(\Phi^t(h_0)\circ T_a)=\frac{d}{dt}\Phi^t(h_0\circ T_a)\nonumber
\end{align}
This means that there exists some $C$ independent of $t$ such that
$[\Phi^t(h_0)]\circ T_a=\Phi^t(h_0\circ T_a)+C$. Take $t=0$. We have
$C=0$, i.e. the first equality holds.\\
For the second equality, we just consider the case when $m=1$ and
observe some symmetry of the gradient flow equation. Due to (H1), we
know that $\Phi^t(h_0)$ is also a solution of
\begin{equation*}
\left\{
  \begin{array}{rl}
   &\frac{d}{dt}(h^t+1)=-X(h^t+1)\\
   &h^0+1=h_0+1.
  \end{array}
\right.
\end{equation*}
This means $\Psi^t(h_0+1)\equiv\Phi^t(h_0)+1$ is a solution of
\begin{equation}\label{ode}
\left\{
  \begin{array}{rl}
   &\frac{d}{dt}(h^t)=-X(h^t)\\
   &h^0=h_0+1.
  \end{array}
\right.
\end{equation}
Moreover, by comparing equation (\ref{ode}) with equation
(\ref{ODEsystem}), we have $\Psi^t(h_0+1)=\Phi^t(h_0+1)$. This
finishes the proof.
\end{proof}
One of the key properties of $\Phi^t$ which was first observed by S.
B. Angenent in the case of standard map (\cite{Angenent'88}) is that
it is strictly monotone, i.e.
\begin{Lemma}[Strong Comparison Principle]\label{lemma:SCP}
If $h,\tilde{h}\in Y$ and $h<\tilde{h}$ and $\omega$ is not
completely resonant, we have
$\Phi^t(h)\prec\!\!\prec\Phi^t(\tilde{h})$ for any $t>0$.
\end{Lemma}
\begin{proof}
We use that the flow in the Banach space is differentiable (see
\cite{Rafael'97} for details). By the general theory of ODE, we also
have that the derivative satisfies the equations of variation
\begin{align}
DX(h)\cdot \eta =\sum_{j=1}^d[&(\partial_{11}H_j(h,h\circ T_{\omega_j})+\partial_{22}H_j(h\circ T_{-\omega_j},h))\cdot\eta+\nonumber\\
                              &\partial_{12}H_j(h\circ T_{-\omega_j},h)\cdot\eta\circ T_{-\omega_j}+\partial_{12}H_j(h,h\circ T_{\omega_j})\cdot\eta\circ
                              T_{\omega_j}]\nonumber\\
\text{ for any }\eta\in L^\infty\nonumber
\end{align}
Let $M^t(h_0)=D\Phi^t(h_0):L^\infty\rightarrow L^\infty$ is a linear
operator which satisfies the operator equation below (often called
variational equation \cite{Hale'80} even if they do not have much to
do with calculus of variations):
\begin{equation*}
\left\{
  \begin{array}{rl}
   &\frac{d}{dt} M^t=-DX(\Phi^t(h_0))\cdot M^t\\
   &M^0=id
  \end{array}
\right.
\end{equation*}
To prove Lemma \ref{lemma:SCP}, due to the fact that $M^t$ is a
linear operator, it suffices to prove that the solution is strictly
positive on $Y$. That is, $0\leq v=\tilde{h}-h$ implies
$0\prec\!\!\prec v^t\equiv M^t(h_0)\cdot v$. In fact, let
$h_0=h+s\cdot(\tilde{h}-h)$ for $0\leq s\leq 1$. By the fundamental
theorem of calculus, we have $\Phi^t(\tilde{h})-\Phi^t(h)=\int_0^1
D\Phi^t(h+s\cdot(\tilde{h}-h))\cdot(\tilde{h}-h)~ds=\int_0^1D\Phi^t(h_0)\cdot
v~ds\geq0$ for any $\theta\in \mathbb{R}$, i.e.
$\Phi^t(h)\prec\!\!\prec\Phi^(\tilde{h})$.

Let $u^t=-\sum_{j=1}^d[\partial_{11}H_j(\Phi^t(h_0),\Phi^t(h_0)\circ
T_{\omega_j})+\partial_{22}H_j(\Phi^t(h_0)\circ
T_{-\omega_j},\Phi^t(h_0))]$ and $W^t=e^{-\int_0^t u^sds}\cdot v^t$.
We get:
\begin{align}
\frac{d}{dt} W^t&=e^{-\int_0^t u^s ds}\cdot u^t\cdot
v^t+e^{-\int_0^t
u^s ds}\cdot\frac{d}{dt} v^t\nonumber\\
        &=-u^t\cdot W^t-\sum_{j=1}^d(\partial_{11}H_j(\Phi^t(h_0),\Phi^t(h_0)\circ T_{\omega_j})\nonumber\\
        &\qquad+\partial_{22}H_j(\Phi^t(h_0)\circ
        T_{-\omega_j},\Phi^t(h_0)))\cdot v^t \cdot e^{-\int_0^t u^s ds} \nonumber\\
        &\qquad-\sum_{j=1}^d\partial_{12}H_j(\Phi^t(h_0)\circ
        T_{-\omega_j},\Phi^t(h_0))\cdot v^t\circ T_{-\omega_j}\cdot
        e^{-\int_0^t u^s ds}\nonumber\\
        &\qquad-\sum_{j=1}^d\partial_{12}H_j(\Phi^t(h_0),\Phi^t(h_0)\circ
        T_{\omega_j})\cdot v^t\circ T_{\omega_j}\cdot e^{-\int_0^t
        u^sds}\nonumber\\
        &=-\sum_{j=1}^d\partial_{12}H_j(\Phi^t(h_0)\circ
        T_{-\omega_j},\Phi^t(h_0))\cdot W^t\circ T_{-\omega_j}\cdot e^{\int_0^t
        (u^s\circ T_{-\omega_j}-u^s)ds}\nonumber\\
        &\quad-\sum_{j=1}^d\partial_{12}H_j(\Phi^t(h_0),\Phi^t(h_0)\circ
        T_{\omega_j})\cdot W^t\circ T_{\omega_j}\cdot e^{\int_0^t
        (u^s\circ T_{\omega_j}-u^s)ds}\nonumber
\end{align}
By using Euler method, for $t$ small enough,
\[W^t=v+t\cdot \frac{d}{dt} W^t\cdot v+O(t^2).\]
Since $0\leq v\in Y$, there exists a small interval $[\alpha,\beta]$
such that $v|_{[\alpha,\beta]}>0$. Since $\omega$ is not completely
resonant, we can find some component $\omega_m$ which is irrational
for some $m\in\{1,\ldots,d\}$. Due to (H2) and Picard's iteration,
$W^{t_1}|_{[\alpha+\omega_j,\beta+\omega_j]}>0$ for sufficiently
small $t_1$ and $j=1,\ldots,d$. In particular,
$W^{t_1}|_{[\alpha+\omega_m,\beta+\omega_m]}>0$. Repeating $k$
times, we get
$W^{t_2}|_{[\alpha+k\cdot\omega_m,\beta+k\cdot\omega_m]}>0$ for
small $t_2>t_1$. Due to the compactness of interval $[0,T]$ and the
fact $\omega_m$ is irrational, this leads to $0\prec\!\!\prec W^t$
for any $t\in(0,T]$. Therefore $0\prec\!\!\prec v^t$ holds for any
$t>0$. This finishes the proof.
\end{proof}

\begin{Proposition}\label{Proposition:2}
$Y$ is invariant under the gradient flow $\Phi^t$, that is,
$\Phi^t(Y)\subseteq Y$ for $t\geq0$.
\end{Proposition}
\begin{proof}
For any $h\in Y$ and $t>0$, since $h<h\circ T_a$ if $a>0$, the fact
that $\Phi^t(h)\prec~\!\!\!\prec \Phi^t(h\circ T_a)$ is just an
immediate consequence of Lemma \ref{lemma:SCP}. We already know that
$\Phi^t(h)\circ T_1=\Phi^t(h)+1$ by Proposition \ref{Proposition:1}.
The left continuity of $\Phi^t(h)$ is from continuity of the
gradient flow with respect the initial data and the definition of
$h$.
\end{proof}

\begin{Lemma}\label{lemma:CP}
If $h^-<h^+$ are both critical points of $\mathscr{P}_\omega$ on
$Y$, then $h^-\prec\!\!\prec h^+$.
\end{Lemma}
\begin{proof}
Due to $h^-<h^+\in Y$ and Lemma \ref{lemma:SCP}, we have
$\Phi^t(h^-)\prec\!\!\prec\Phi^t(h^+)$. On the other hand, since
$h^-$ and $h^+$ are both critical points of $\mathscr{P}_\omega$,
$\Phi^t(h^-)=h^-$ and $\Phi^t(h^+)=h^+$ hold by Lemma
\ref{lemma:gf}. This finishes the proof.
\end{proof}

\begin{proof}[Proof of Theorem \ref{Theorem:Main2}]
(1) is an immediate consequence of Lemma \ref{lemma:CP}.

In order to prove (2), we follow the method used by
\cite{Rafael'07a}. We define the compact set $\mathcal
{K}\equiv\{h\in Y:h^-\leq h\leq h^+\}$.
Due to the compactness of $\mathcal{K}$, the topology induced by
$L^\infty$ norm and the topology induced by the Hausdorff metric are
equivalent on $\mathcal{K}$. For any $h\in\mathcal{K}$, we know that
$h^-=\Phi^t(h^-)\leq\Phi^t(h)\leq\Phi^t(h^+)=h^+$ by Lemma
\ref{lemma:SCP} and the definition of $h^-$ and $h^+$. This means
that $\Phi^t(\mathcal{K})\subseteq\mathcal{K}$ due to Proposition
\ref{Proposition:2}. Let $h^s=s\cdot h^++(1-s)\cdot h^-$ for any
$s\in[0,1]$. We have
\begin{equation}\label{percivallagrangeunderflow}
\frac{d}{dt}\mathscr{P}_\omega(\Phi^t(h^s))=-\int_0^1
|X(\Phi^t(h^s))|^2d\theta\leq0,
\end{equation}
i.e. $\mathscr{P}_\omega(\Phi^t(h^s))$ is decreasing with respect to
$t$ for any fixed $s\in[0,1]$. Since
$\mathscr{P}_\omega|_\mathcal{K}$ is bounded and
$\frac{d^2}{dt^2}\mathscr{P}_\omega(\Phi^t(h^s))$ is bounded from
above,
$\lim_{t\rightarrow\infty}\mathscr{P}_\omega(\Phi^t(h^s))$ exists and $\lim_{t\rightarrow\infty}\frac{d}{dt}\mathscr{P}_\omega(\Phi^t(h^s))=0$.\\
Let
\[
\mathscr{B}_\omega=\max_{s\in[0,1]}\inf_{t\geq0}\mathscr{P}_\omega(\Phi^t(h^s))\equiv\max_{s\in[0,1]}\lim_{t\rightarrow\infty}\mathscr{P}_\omega(\Phi^t(h^s))\geq
\mathscr{P}_\omega(h^-).
\]

There are two possibilities
$\mathscr{B}_\omega>\mathscr{P}_\omega(h^-)$ or
$\mathscr{B}_\omega=\mathscr{P}_\omega(h^-)$. We will show that the
conclusion holds in each of the two cases.

\begin{itemize}
\item If $\mathscr{B}_\omega>\mathscr{P}_\omega(h^-)$, there exists
$s_0\in(0,1)$ such that
\[
\lim_{t\rightarrow\infty}\mathscr{P}(\Phi^t(h^{s_0}))=\mathscr{B}_\omega
\]
and
\[
\lim_{t\rightarrow\infty}\frac{d}{dt}\mathscr{P}(\Phi^t(h^{s_0}))=0.
\]
Due to the compactness of $\mathcal{K}$, we can extract a
subsequence $t_n\rightarrow\infty$ such that
$\Phi^{t_n}(h^{s_0})\rightarrow h^*\in\mathcal{K}$. This leads to
$\mathscr{P}_\omega(h^*)=\mathscr{B}_\omega$ which means that $h^*$
is different from $h^-$ and $h^+$. In the other hand, due to
(\ref{percivallagrangeunderflow}), we have
$\lim_{t_n\rightarrow\infty}\int_0^1|X(\Phi^{t_n}(h^{s_0}))|^2d\theta=\int_0^1|X(h^*)|^2d\theta=0$.
Since $h^*$ is left-continuous, we get $X(h^*)=0$ which means $h^*$
is a critical point of $\mathscr{P}_\omega$. This finishes the proof
when $\mathscr{B}_\omega>\mathscr{P}_\omega(h^-)$.

\item If $\mathscr{B}_\omega=\mathscr{P}_\omega(h^-)$, we have
$\inf_{t\geq0}\mathscr{P}(\Phi^t(h^s))\leq
\mathscr{B}_\omega=\mathscr{P}_\omega(h^-)$. This means
$\inf_{t\geq0}\mathscr{P}(\Phi^t(h^s))=\mathscr{P}_\omega(h^-)$ for
any $s\in[0,1]$.  We now argue by contradiction and assume that no
other critical point (and so a fortiori no minimizer) but $h^-$ and
$h^+$ in $\mathcal{K}$. We have two alternatives, both of which lead
to contradictions with the non-existence of other critical points.

\begin{enumerate}
\item[(a)] One is that the omega limit set of $\Phi^t(h^s)$ contains
both $\{h^-,h^+\}$. Let
\begin{align}
&B_r(h^-)\equiv \{ h\in\mathcal{K}:d(h,h^-)<r \},\nonumber\\
&B_r(h^+)\equiv \{ h\in\mathcal{K}:d(h,h^+)<r \}\nonumber
\end{align} denote
the $r$-ball of $h^-$ and $h^+$ respectively in $\mathcal{K}$ . Take
$0<r<\frac{1}{2}d(h^-,h^+)$ sufficiently small such that
$B_r(h^-)\cap B_r(h^+)=\phi$. Thus there exists an $M_0(r)$ in this
case, such that $\Phi^t(h^s)\in B_r(h^-)\cup B_r(h^+)$ for any
$t>M_0(r)$. Let $D^-\equiv \{t>M_0(r):\Phi^t(h^s)\in B_r(h^-)\}$ and
$D^+\equiv \{t>M_0(r):\Phi^t(h^s)\in B_r(h^+)\}$ which are nonempty.
We know $D^-\cap D^+=\phi$. By the continuity of $\Phi^t(h^s)$ with
respect to $t$ these two sets are open. This means that two nonempty
disjoint open sets $D^-$ and $D^+$ cover a connected open interval
$(M_0(r),\infty)$ which is a contradiction.
\item[(b)] The other is that the omega limit set of $\Phi^t(h^s)$ has
only one point either $\{h^-\}$ or $\{h^+\}$. Let $E^-\equiv \{s\in
(0,1):\lim_{t\rightarrow \infty}\Phi^t(h^s)=h^-\}$ and $E^+\equiv
\{s\in (0,1):\lim_{t\rightarrow \infty}\Phi^t(h^s)=h^+\}$ which are
nonempty open sets due to the continuous dependence of $\Phi^t$ on
initial data. This is a contradiction by the same trick used in (a).
\end{enumerate}
\end{itemize}
This completes the proof of (2).
\end{proof}


\section*{Appendix A. Hull function approach to general lattices}
The method of hull functions can be extended to more general
lattices.

For simplicity, we discuss only when the place of $\mathbb{Z}^d$ is
taken by a finitely generated group $G$ and the interaction is
invariant under the action of $G$, as well as by addition of $1$ to
the configurations. See \cite{Rafael'10} for more general lattices.

Because of the translation invariance, we consider variational
principles
\begin{equation}\label{formalvp}
\mathscr{L}(u)=\sum_{\substack{B\subseteq G\\
\sharp B~~ \text{finite}\\0\in B}} S_{B\cdot g}(u)
\end{equation}where $S_B$ depends only on $u|_B$.

We recall that $\omega:~G\rightarrow \mathbb{R}$ is a cocycle when
$\omega(g\cdot \tilde{g})=\omega(g)+\omega(\tilde{g})$.

Given a cocycle $\omega$ we seek configurations:
\begin{equation}\label{generalhull}
x_g=h_\omega(\omega\cdot g)
\end{equation} for $h_\omega:~\mathbb{R}\rightarrow \mathbb{R}$
monotone, $h_\omega(t+1)=h_\omega(t)+1$.

It is immediate that all configurations \eqref{generalhull} satisfy
for all $k,~g\in G,~l\in\mathbb{Z}$
\begin{equation}\label{omegaBirkhoff}
x_{g\cdot k}+l\leq x_g\Longleftrightarrow\omega(k)+l\leq 0
\end{equation}which is an analogue of the $\omega$-Birkhoff
property. Similarly, one can easily see that the $\omega$-Birkhoff
property \eqref{omegaBirkhoff} implies the existence of a hull
function.

Given $B=\{s_0=0,s_1,\ldots,s_n\}\subseteq G$ we can write
$S_B(u)=S_B(u_0,u_{s_1},\ldots,u_{s_n})$. Given a variational
principle \eqref{formalvp} we can associate the Percival variational
principle
\begin{equation}\label{percival}
\mathscr{P}_\omega(h)=\int_0^1d\theta\sum_{\substack{B\subseteq
G\\\sharp B~\text{finite}\\0\in
B}}S_B(h(\theta),h(\theta+\omega(s_1)),\ldots,h(\theta+\omega(s_n))).
\end{equation}

We use the same procedure as the commutative group case
($\mathbb{Z}^d$) to prove the existence of the minimal
configurations generated by hull functions approach. Namely:
\begin{itemize}
\item The minimizers (resp. critical points) of \eqref{formalvp}
give via \eqref{generalhull} class-$A$ (resp. critical)
configurations.

\item For every cocycle $\omega$, there exists a class-$A$
minimizer.

\item For every $\omega$ there are at least two different critical
points. If there are two minimizers, then one gets a circle of
critical points.
\end{itemize}
 It is easy to
get the following theorem.
\begin{Theorem}
Under the assumptions as in \cite{Rafael'98} for general lattices,
there is a minimizer $h_\omega$ of $\mathscr{P}_\omega$ over $Y$ or
$Y^*$. $x_g=h_\omega(\omega\cdot g)$ is a $\omega$-Birkhoff ground
state of cocycle $\omega$. In addition, if both $h^- < h^+$ are
minimizers of $\mathscr{P}_\omega$ on $Y$ then $h^-\prec\!\!\prec
h^+$ and there is a critical point in between.
\end{Theorem}

\begin{proof}
We give a sketch of proof. We assume that the sum
\[
\sum_{\substack{B\subseteq G\\\sharp B~\text{finite}\\0\in
B}}S_B(h(\theta),h(\theta+\omega(s_1)),\ldots,h(\theta+\omega(s_n))).
\]converges uniformly. We first check the symmetries and obtain
\[
\mathscr{P}_\omega(h\circ T_a)=\sum_{\substack{B\subseteq G\\\sharp
B~\text{finite}\\0\in
B}}\int_0^1S_B(h(\theta+a),h(\theta+a+\omega(s_1)),\ldots,h(\theta+a+\omega(s_n)))=\mathscr{P}_\omega(h),
\]and
\[
\mathscr{P}_\omega(h+1)=\int_0^1d\theta\sum_{\substack{B\subseteq
G\\\sharp B~\text{finite}\\0\in
B}}S_B(h(\theta)+1,h(\theta+\omega(s_1))+1,\ldots,h(\theta+\omega(s_n))+1)=\mathscr{P}_\omega(h).
\]
In addition, we assume that $S_B$ satisfies the weak twist condition
(see \cite{Rafael'98})
\[
\sum_{B\ni q}\frac{\partial^2}{\partial p\partial q}S_B(u)\leq 0
\]for any $p\neq q$. The twist condition implies the rearrangement
inequality:
\[
\mathscr{P}(h\wedge\tilde{h})+\mathscr{P}(h\vee\tilde{h})\leq
\mathscr{P}(h)+\mathscr{P}(\tilde{h}).
\]

\end{proof}

\begin{Remark}
The heuristic argument for \eqref{percival} is that, even if the
group $G$ is not amenable since we consider only configurations
which depend only on the value of the cocycle and which transform
well, we only need to average over the values of the cocycle.

One can make assumptions that argue that the sum \eqref{formalvp}
converges. For example $S_B=0$ when $dian_B\geq R$ (finite range).
We have not explored what are the optimal assumptions.
\end{Remark}


\section*{Acknowledgements}
The work of R. L. has been supported by NSF grant DMS 0901389. The
visit of X. S. to U.T. Austin has been sponsored by CSC grant
2003619040. Both authors thank CRM (Barcelona) for support during
Fall 2008. X. S. thank U. of Texas Austin for hospitality. We also
thank Prof. C.Q. Cheng for encouragement and support. X. S thank
dynamical system group at Nanjing University such as Prof. W. Cheng,
X. Cui etc.
\bibliographystyle{alpha}
\bibliography{reference}

\begin{thebibliography}{KdlLR97}

\bibitem[ALD83]{Aubry'83}
S.~Aubry and P.~Y. Le~Daeron.
\newblock The discrete {F}renkel-{K}ontorova model and its extensions. {I}.
  {E}xact results for the ground-states.
\newblock {\em Phys. D}, 8(3):381--422, 1983.

\bibitem[Ang88]{Angenent'88}
S.~B. Angenent.
\newblock The periodic orbits of an area preserving twist map.
\newblock {\em Comm. Math. Phys.}, 115(3):353--374, 1988.

\bibitem[Ban88]{Bangert'88}
V.~Bangert.
\newblock Mather sets for twist maps and geodesics on tori.
\newblock In {\em Dynamics reported, {V}ol.\ 1}, volume~1 of {\em Dynam.
  Report. Ser. Dynam. Systems Appl.}, pages 1--56. Wiley, Chichester, 1988.

\bibitem[Ban89]{Bangert'89}
V.~Bangert.
\newblock On minimal laminations of the torus.
\newblock {\em Ann. Inst. H. Poincar\'e Anal. Non Lin\'eaire}, 6(2):95--138,
  1989.

\bibitem[BK04]{MR2035039}
O.~M. Braun and Y.~S. Kivshar.
\newblock {\em The {F}renkel-{K}ontorova model}.
\newblock Texts and Monographs in Physics. Springer-Verlag, Berlin, 2004.
\newblock Concepts, methods, and applications.

\bibitem[CdlL98]{Rafael'98}
A.~Candel and R.~de~la Llave.
\newblock On the {A}ubry-{M}ather theory in statistical mechanics.
\newblock {\em Comm. Math. Phys.}, 192(3):649--669, 1998.

\bibitem[CdlL09]{MR2507323}
Renato Calleja and Rafael de~la Llave.
\newblock Fast numerical computation of quasi-periodic equilibrium states in
  1{D} statistical mechanics, including twist maps.
\newblock {\em Nonlinearity}, 22(6):1311--1336, 2009.

\bibitem[CI99]{Contreras'99}
Gonzalo Contreras and Renato Iturriaga.
\newblock {\em Global minimizers of autonomous {L}agrangians}.
\newblock 22$^{\rm o}$ Col\'oquio Brasileiro de Matem\'atica. [22nd Brazilian
  Mathematics Colloquium]. Instituto de Matem\'atica Pura e Aplicada (IMPA),
  Rio de Janeiro, 1999.

\bibitem[dlL08]{Rafael'08}
Rafael de~la Llave.
\newblock K{AM} theory for equilibrium states in 1-{D} statistical mechanics
  models.
\newblock {\em Ann. Henri Poincar\'e}, 9(5):835--880, 2008.

\bibitem[dlLV07a]{Rafael'07a}
Rafael de~la Llave and Enrico Valdinoci.
\newblock Ground states and critical points for generalized
  {F}renkel-{K}ontorova models in {$\Bbb Z\sp d$}.
\newblock {\em Nonlinearity}, 20(10):2409--2424, 2007.

\bibitem[dlLV07b]{LlaveVPDE}
Rafael de~la Llave and Enrico Valdinoci.
\newblock Multiplicity results for interfaces of
  {G}inzburg-{L}andau-{A}llen-{C}ahn equations in periodic media.
\newblock {\em Adv. Math.}, 215(1):379--426, 2007.

\bibitem[dlLV10]{Rafael'10}
Rafael de~la Llave and Enrico Valdinoci.
\newblock Ground states and critical points for {A}ubry-{M}ather theory in
  statistical mechanics.
\newblock {\em J. Nonlinear Sci.}, 20(2):153--218, 2010.

\bibitem[Ell69]{Ellis}
Robert Ellis.
\newblock {\em Lectures on topological dynamics}.
\newblock W. A. Benjamin, Inc., New York, 1969.

\bibitem[Fat97]{Fathi'97}
Albert Fathi.
\newblock Th\'eor\`eme {KAM} faible et th\'eorie de {M}ather sur les syst\`emes
  lagrangiens.
\newblock {\em C. R. Acad. Sci. Paris S\'er. I Math.}, 324(9):1043--1046, 1997.

\bibitem[Fig08]{Figalli}
Alessio Figalli.
\newblock {\em {Optimal transportation and action-minimizing measures.}}
\newblock PhD thesis, {Tesi. Scuola Normale Superiore Pisa (Nuova Serie) 8.
  Pisa: Edizioni della Normale; Pisa: Scuola Normale Superiore (Thesis). 254~p.
  EUR~19.26 }, 2008.

\bibitem[For96]{Forni'96}
Giovanni Forni.
\newblock Construction of invariant measures supported within the gaps of
  {A}ubry-{M}ather sets.
\newblock {\em Ergodic Theory Dynam. Systems}, 16(1):51--86, 1996.

\bibitem[FV11]{FarinaV}
Alberto Farina and Enrico Valdinoci.
\newblock Some results on minimizers and stable solutions of a variational
  problem.
\newblock 2011.
\newblock To appear in {E}rgodic {T}heory and {D}ynamical {S}ystems.

\bibitem[GHM08]{Monneau'08}
M.-A. Ghorbel, P.~Hoch, and R.~Monneau.
\newblock A numerical study for the homogenisation of one-dimensional models
  describing the motion of discolations.
\newblock {\em Int. J. Comput. Sci. Math.}, 2(1-2):28--52, 2008.

\bibitem[Gol01]{Gole'01}
Christophe Gol{\'e}.
\newblock {\em Symplectic twist maps}, volume~18 of {\em Advanced Series in
  Nonlinear Dynamics}.
\newblock World Scientific Publishing Co. Inc., River Edge, NJ, 2001.
\newblock Global variational techniques.

\bibitem[Hal80]{Hale'80}
Jack~K. Hale.
\newblock {\em Ordinary differential equations}.
\newblock Robert E. Krieger Publishing Co. Inc., Huntington, N.Y., second
  edition, 1980.

\bibitem[Hed32]{MR1503086}
Gustav~A. Hedlund.
\newblock Geodesics on a two-dimensional {R}iemannian manifold with periodic
  coefficients.
\newblock {\em Ann. of Math. (2)}, 33(4):719--739, 1932.

\bibitem[JGV09]{JungingerV}
Hannes Junginger-Gestrich and Enrico Valdinoci.
\newblock Some connections between results and problems of {D}e {G}iorgi,
  {M}oser and {B}angert.
\newblock {\em Z. Angew. Math. Phys.}, 60(3):393--401, 2009.

\bibitem[Kat83]{Katok}
A.~Katok.
\newblock Periodic and quasiperiodic orbits for twist maps.
\newblock In {\em Dynamical systems and chaos ({S}itges/{B}arcelona, 1982)},
  volume 179 of {\em Lecture Notes in Phys.}, pages 47--65. Springer, Berlin,
  1983.

\bibitem[KdlLR97]{Rafael'97}
Hans Koch, Rafael de~la Llave, and Charles Radin.
\newblock Aubry-{M}ather theory for functions on lattices.
\newblock {\em Discrete Contin. Dynam. Systems}, 3(1):135--151, 1997.

\bibitem[Kra96]{Kra'96}
Bryna Kra.
\newblock The conjugating map for commutative groups of circle diffeomorphisms.
\newblock {\em Israel J. Math.}, 93:303--316, 1996.

\bibitem[LM01]{LM'01}
M.~Levi and J.~Moser.
\newblock A {L}agrangian proof of the invariant curve theorem for twist
  mappings.
\newblock In {\em Smooth ergodic theory and its applications ({S}eattle, {WA},
  1999)}, volume~69 of {\em Proc. Sympos. Pure Math.}, pages 733--746. Amer.
  Math. Soc., Providence, RI, 2001.

\bibitem[Ma{\~n}91]{Manebook}
Ricardo Ma{\~n}{\'e}.
\newblock Global variational methods in conservative dynamics.
\newblock 18$^{\rm o}$ Col\'oquio Brasileiro de Matem\'atica. [22nd Brazilian
  Mathematics Colloquium], page 170. Instituto de Matem\'atica Pura e Aplicada
  (IMPA), Rio de Janeiro, 1991.

\bibitem[Ma{\~n}96a]{Mane'96}
Ricardo Ma{\~n}{\'e}.
\newblock Generic properties and problems of minimizing measures of
  {L}agrangian systems.
\newblock {\em Nonlinearity}, 9(2):273--310, 1996.

\bibitem[Ma{\~n}96b]{Mane'96b}
Ricardo Ma{\~n}{\'e}.
\newblock Lagrangian flows: the dynamics of globally minimizing orbits.
\newblock In {\em International {C}onference on {D}ynamical {S}ystems
  ({M}ontevideo, 1995)}, volume 362 of {\em Pitman Res. Notes Math. Ser.},
  pages 120--131. Longman, Harlow, 1996.

\bibitem[Mat82a]{Mather'82}
John~N. Mather.
\newblock Existence of quasiperiodic orbits for twist homeomorphisms of the
  annulus.
\newblock {\em Topology}, 21(4):457--467, 1982.

\bibitem[Mat82b]{Mather'82b}
John~N. Mather.
\newblock Nonuniqueness of solutions of {P}ercival's {E}uler-{L}agrange
  equation.
\newblock {\em Comm. Math. Phys.}, 86(4):465--473, 1982.

\bibitem[Mat85]{Mather'85}
John~N. Mather.
\newblock More {D}enjoy minimal sets for area preserving diffeomorphisms.
\newblock {\em Comment. Math. Helv.}, 60(4):508--557, 1985.

\bibitem[Mat86]{Mather'86}
John Mather.
\newblock A criterion for the nonexistence of invariant circles.
\newblock {\em Inst. Hautes \'Etudes Sci. Publ. Math.}, (63):153--204, 1986.

\bibitem[Mat89]{Mather'89}
John~N. Mather.
\newblock Minimal measures.
\newblock {\em Comment. Math. Helv.}, 64(3):375--394, 1989.

\bibitem[Mat91]{Mather'91}
John~N. Mather.
\newblock Action minimizing invariant measures for positive definite
  {L}agrangian systems.
\newblock {\em Math. Z.}, 207(2):169--207, 1991.

\bibitem[Mor24]{Morse'24}
Harold~Marston Morse.
\newblock A fundamental class of geodesics on any closed surface of genus
  greater than one.
\newblock {\em Trans. Amer. Math. Soc.}, 26(1):25--6 0, 1924.

\bibitem[Mor73]{Morse'73}
Marston Morse.
\newblock {\em Variational analysis: critical extremals and {S}turmian
  extensions}.
\newblock Interscience Publishers [John Wiley \& Sons, Inc.], New
  York-London-Sydney, 1973.
\newblock Pure and Applied Mathematics.

\bibitem[Mos86]{Moser'86}
J{\"u}rgen Moser.
\newblock Minimal solutions of variational problems on a torus.
\newblock {\em Ann. Inst. H. Poincar\'e Anal. Non Lin\'eaire}, 3(3):229--272,
  1986.

\bibitem[Pal79]{Palais'79}
Richard~S. Palais.
\newblock The principle of symmetric criticality.
\newblock {\em Comm. Math. Phys.}, 69(1):19--30, 1979.

\bibitem[Per79]{Percival'79}
I.~C. Percival.
\newblock A variational principle for invariant tori of fixed frequency.
\newblock {\em J. Phys. A}, 12(3):L57--L60, 1979.

\bibitem[Poi85]{Poincare}
H.~Poincar\'e.
\newblock Sur les courbes d\'efinies par les \'equations diff\'erentielles.
\newblock {\em J. Math Pures et Appl.}, 1:167--244, 1885.

\bibitem[SdlL11]{SuL11}
Xifeng Su and Rafael de~la Llave.
\newblock K{AM} theory for quasi-periodic equilibria in 1-{D} quasiperiodic
  media.
\newblock 2011.
\newblock In preparation.

\bibitem[SZ89]{MR982563}
Dietmar Salamon and Eduard Zehnder.
\newblock K{AM} theory in configuration space.
\newblock {\em Comment. Math. Helv.}, 64(1):84--132, 1989.

\end{thebibliography}
\end{document}